\documentclass[review]{autart}

\usepackage{graphicx}
\usepackage{cite}
\usepackage{picinpar}
\usepackage{flushend}
\usepackage{colortbl}
\usepackage{soul}
\usepackage{multirow}
\usepackage{pifont}
\usepackage{color}
\usepackage{alltt}
\usepackage{enumerate}
\usepackage{siunitx}
\usepackage{pbox}
\usepackage{mathtools}
\usepackage{accents}
\usepackage{steinmetz}
\usepackage{amsmath,amssymb,amsfonts}
\usepackage{xcolor}
\usepackage{xfrac}
\usepackage{units}
\usepackage{algorithmic}
\usepackage[english]{babel}  
\usepackage{tabularx}
\usepackage{epstopdf}
\usepackage{textcomp}
\usepackage{caption}
\usepackage{arydshln}

\usepackage[mathscr]{euscript}
\let\euscr\mathscr \let\mathscr\relax
\usepackage[scr]{rsfso}
\newtheorem{proof}{Proof}

\newcommand*\bigstrut{%
	\vrule height\baselineskip depth.5\baselineskip width 0pt\relax}

\journal{Journal of \LaTeX\ Templates}

\begin{document}

\begin{frontmatter}

\title{Modeling, Analysis, and Control of Mechanical Systems under Power Constraints}

\author{Gorkem Secer}

\address{Control Systems Design Dept., ROKETSAN Missile Industries Inc., Ankara, Turkey.}

\begin{abstract}
Significant improvements have been achieved in
motion control systems with the availability of high speed
power switches and microcomputers on the market. Even
though motor drivers are able to provide high torque control
bandwidth under nominal conditions, they suffer from various
physical constraints which degrade both output amplitude and
bandwidth of torque control loop. In this context, peak power 
limit of a power source, as one of those constraints, has not 
been fully explored from the control perspective so far. 
A conventional and practical way of considering peak power limit
in control systems is to model it as a trivial torque saturation 
derived from the allowable torque at maximum speed satisfying the 
constraint. However, this model is overly conservative leading to 
poor closed loop performance when actuators operate below their 
maximum speed. In this paper, novel ways of incorporating peak 
power limits into both classical and optimal controllers are 
presented upon a theoretical analysis revealing its effects 
on stability and performance. 
\end{abstract}


\end{frontmatter}

\section{Introduction}
\label{sec:intro}
Since operational lifetime of mobile systems is limited by capacity of their onboard power supply and by energy expenditure of actuators, running time can only be maximized by combined improvements in both supply side and actuation. In this context, for the power supply, while some systems use energy harvesting mechanisms \cite{melhuish.CONF2000,landis.PVSC2005,manning.SYSoSE2013}, there is an interesting work \cite{hutter.CONF2017} which adds autonomous recharging capabilities to a mobile robot for prolonged service. From the viewpoint of actuation, energy consumption can be minimized either by seeking optimal mechanical design \cite{avik_koditschek.RAL2016,hurst_hatton.RSS2016,hatton_hurst.IEEE2018} for a particular behavior or by designing energy efficient controls \cite{collins_ruina_tedrake_wisse.SCIENCE2005,secer_saranli.TRO2018}. Focus of this paper is on the latter, particularly on motion controllers respecting instantaneous power supply limits.

Motion control systems consist of actuators and drivers.
While the actuator generates torques/forces necessary to move the
system and/or payload attached to it, the driver is responsible for controlling the actuator. In this context, electrical
motor-based systems are becoming extremely popular thanks to improvements in magnet materials (e.g., NdFeB) and
inverter technology. Permanent-magnet synchronous motors
(PMSMs) are highly preferred within this trend compared to
other AC and DC machines due to their efficiency and
torque density \cite{sangbae.IROS2012}. Drivers control PMSMs by converting bus voltages/currents to motor voltages/currents with ratios adjustable via duty cycles of switching pulse width modulation signals. Developments in high speed digital devices and switching circuits have increased the performance of torque controllers on drivers such that relatively high bandwidth can be provided allowing motion control designers to ignore electrical dynamics and inefficiencies. Based on this abstraction, it is often sufficient to model the driver
and motor combination as a perfect torque source. However,
motion control systems suffer from various nonlinearities
due to physical limitations in driving circuit and power
supply. There are two main nonlinearities present in motion
controllers in this context: \\
I) Peak current rating of switches (e.g., MOSFETs) which can be modeled as saturation of motor torques $u \in \mathbb{R}^n$ in the form
\begin{equation}
{\rm sat} (u, u_{\max}) = [{\rm sat}(u_1) \quad {\rm sat}(u_2) \quad\cdots\quad {\rm sat}(u_n)]^T
\label{eq:sat}
\end{equation}
with ${\rm sat}(u_i) = {\rm sign}(u_i) \min\{\lvert u_i \rvert,u_{\max}\}$, $i=1, 2, \dots n$. Due to saturation, actual and desired controller output become different, thus leading controller states to be updated wrongly. This may result in large overshoots at plant output and sometimes even instability (see \cite{doyle.ACC1987,choi_lee.TIE2009}) when the controller has slow and/or unstable modes (i.e., poles and zeros close to and/or on the right hand side of imaginary axis, respectively). Numerous techniques and extensions to controllers, such as conditional integration \cite{ross.ISA1967,hanus.AUTO1987}, reference governors \cite{kolmanovsky.AUTO2002, kolmanovsky.ACC2014}, variable structure controllers \cite{hall.TIE2001} or anti windup with model recovery filter \cite{zaccarian.EJC2009}, have been proposed in the literature to handle saturation effects in terms of both stability and performance. Most of these techniques have provable stability results which are generally obtained by describing function analysis of saturation nonlinearity \cite{astrom.ACC1989}, linear matrix inequalities \cite{teel_zaccarian.TAC2003,morari.AUTO2001} or convex programming \cite{tarbouriech.CDC1997,boyd.CDC1998}. 

\noindent II) Power limit of the onboard power supply batteries and circuits \cite{forsyth.TIE2015}, which is our focus in this paper. Even
though this constraint takes place in electrical part of the
system, it can be mapped to motion control level by encoding the constraint into software as
\begin{equation*}
{\rm psat} (u, \dot{q}) = [{\rm psat}(u_1, \dot{q}_1, \bar{P}_1) \;\;\cdots\;\; {\rm psat}(u_n, \dot{q}_n, \bar{P}_n)]^T
\end{equation*}where {\vspace{-0.8cm}\small
\begin{equation*}{\rm psat} (u_i, \dot{q}_i, \bar{P}_i) = \begin{cases} u_i & \text{if $P_i \leq \bar{P}_i$} \\
\dfrac{-\dot{q}_i \pm \sqrt{\strut\dot{q}_i^2 + 4 \bar{P}_i R_i / (k_t)_i^2}}{2 R_i / (k_t)_i^2} & \text{else if $u_i \gtrless 0$},
\end{cases}
\end{equation*}}$\dot{q}_i$ is the generalized velocity of the $i^{th}$ joint, $\bar{P}_i$ is the power budget allocated to $i^{th}$ joint supplied by a common source having peak power $P_{\max}$ with $P_{\max} = \sum_{i=1}^{n} \bar{P}_i$, and the power input $P_i$ to the $i^{th}$ motor includes mechanical power output determined by torque constant $(k_t)_i$ and motor losses dissipated across the resistance $R_i$ as
\begin{equation}
P_i:=u_i \dot{q}_i + u_i^2 R_i/(k_t)_i^2.
\label{eq:motorpower_full}
\end{equation}
For high efficiency motors, electrical losses can be ignored which leads to a simpler expression for power limit saturation as
\begin{equation}
{\rm psat} (u_i, \dot{q}_i) = \begin{cases} u_i & \text{if $P_i \leq \bar{P}_i$} \\ \bar{P}_i/\dot{q}_i & \text{otherwise.}
\end{cases}
\label{eq:psat_simple}
\end{equation}

The power limit is a nonlinearity, which might lead to instability and deterioration in closed loop performance. In contrary to the extensive amount of work on torque saturation, we have not found any work considering the external power supply limit constraint explicitly to this date, despite the fact that it has become particularly important for actuators with recent advances in mobile systems. A relatively reasonable yet conservative approximation that might be used in practice is
\begin{equation}
{\rm psat}(u_i, \dot{q}_i, \bar{P}_i) \approx {\rm sat}(u_i, \dfrac{\bar{P}_i}{\bar{v}_i})
\label{eq:psat_approx}
\end{equation}
where $\bar{v}_i$ denotes the no-load velocity of a motor located at $i^{th}$ joint. However, this may result in superfluous actuator and power supply selections to meet performance specifications of the motion controller, eventually limiting the overall system performance especially in the applications such as legged robots or electric vehicles where total weight is crucial for energetic efficiency. To show that, exact model of the power supply limit nonlinearity is analyzed in this paper from various perspectives in comparison with the approximate model (\ref{eq:psat_approx}). In particular, we derive the describing function of the power limit nonlinearity and maximum closed loop frequency attainable in its presence to understand gain and phase effects on stability and steady-state performance, respectively. This analysis is then used to motivate design of stable controllers based on exact model of the nonlinearity
for Euler-Lagrange mechanical systems. These controllers require that the total power limit of the battery is allocated to individual power budget of each actuator unit (i.e., $\bar{P}_i$) for which a static solution strategy (i.e., assigning a fixed power limit prior to operation) is used conventionally. Note that, the domain of power saturation function can be reduced to
\begin{equation*}
{\rm psat} (u, \dot{q}) = [{\rm psat}(u_1, \dot{q}_1) \;\;\cdots\;\; {\rm psat}(u_n, \dot{q}_n)]^T.
\end{equation*}
when the static allocation is employed. On the other hand, performance can be improved by updating power budget dynamically with time-varying limits $\bar{P}_i(t, q_i,\dot{q}_i)$. For example, when some joints consume less power than their limits, channelling their leftover power to remaining joints might improve the performance. To maximize this advantage, a general control architecture combining dynamic power allocation with exact power limit model is developed for linear Euler-Lagrange mechanical systems in this paper. The proposed architecture is posed as an optimization problem to solve a finite horizon optimal control problem and dynamic power allocation simultaneously. Note that, unless otherwise specified, the static allocation scheme is used throughout the paper.

The organization of this paper is as follows : Section~\ref{sec:frequencydomain} presents frequency domain analysis of the power limit nonlinearity that provides motivation to design controllers accounting for this nonlinearity. After such classical controllers are proposed in Section~\ref{sec:controller}, they are extended to a new class of algorithms in Section~\ref{sec:mpc} such that dynamic power allocation and feedback control is unified under a general architecture. Section~\ref{sec:experiments} gives preliminary experimental results on a single degree-of-freedom (DOF) actuator prototype. Finally, conclusions and future work are presented in Sec.~\ref{sec:conclusion}.

\section{Frequency Response Characteristics} 
\label{sec:frequencydomain}
In this section, effects of power limit $\rm{psat}$ (\ref{eq:psat_simple}) (i.e., the simpler model without electrical losses) on closed-loop stability and performance are investigated in comparison to the approximate model (\ref{eq:psat_approx}). In particular, the sinusoidal input describing function of nonlinearities (i.e., a first harmonic approximation to the nonlinearity) are derived to study destabilizing characteristic of both phenomenons, whereas the maximum closed loop bandwidth is derived to determine a benchmark for linear controllers using different strategies for power supply limit. Findings will be used as foundational premises to motivate more in-depth exploration of using exact models of power supply limit for controller design in the subsequent section.

\subsection{Describing Function Analysis}
\label{sec:describingfunction}
Describing function of a nonlinearity is a frequency domain approximation known as quasi-linerization \cite{gelb.BOOK1968}. Even though it was introduced long ago, it is still widely used to predict oscillations of linear dynamical systems subjected to a nonlinearity \cite{li.IEEE2018,boiko.IJC2009,nien.TPOWELEC2016}. In this context, the nonlinearity is approximated by its fundamental harmonic frequency response under a sinusoidal input to identify
limit cycles and study their local behavior. This approximation often provides sufficient accuracy for assessment of periodic behaviors, modeled by frequency domain response, in terms of stability and gain/phase characteristics since linear systems attenuates higher order harmonics considerably due to their low-pass filter characteristics \cite{fridman.AUTO2019}. In particular, the main focus of this paper is on motion control of electromechanical systems whose dynamics have at least $2^{nd}$ order relative degree, hence filtering out higher-order harmonics to a large extent.

Consider a single DoF actuator whose transfer function from torque $u$ to velocity $\dot{q}$ is denoted by $G(s)$. Denoting the describing function of the power saturation by $N(A,\omega)$, transfer function and describing function can be expressed in phasor notation as 
\begin{equation*}
\begin{aligned}
G(s)&=X_G(\omega)~\phase{\phi_G(\omega)} \\
N(A,\omega)&=X_N(A,\omega)~\phase{\phi_N(A,\omega)}.
\end{aligned}
\end{equation*}
Cascade connection of the describing function and transfer function is defined by
\begin{equation*}
G(s) N(A, \omega) = X(A,\omega) \phase{\phi(A,\omega)}
\end{equation*}
with 
\begin{equation*}
\begin{aligned}
X(A,\omega) &= X_G(\omega) X_N(A,\omega) \\
\phi(A,\omega) &= \phi_G(\omega) +  \phi_N(A,\omega).
\end{aligned}
\end{equation*}

First, consider the case where power saturation is not active, i.e., $P\leq P_{\max}$ as defined in (\ref{eq:psat_simple}). Since $u~=~{\rm psat}(u, \dot{q})$, the describing function becomes $N(A,\omega) = 1$. For the other case with $P > P_{\max}$, we start the derivation of describing function by assuming a sinusoidal torque input $u = A \sin(\omega t)$ 
In this case, the velocity and required power take the form
\begin{equation*}
\begin{aligned}
\dot{q} &= A X \sin(\psi + \phi) \\
P &= A^2 X \sin(\psi + \phi)\sin \psi = \dfrac{A^2 X}{2} \left[\cos\phi - \cos(2\psi+\phi) \right]
\end{aligned}
\end{equation*}
with $\psi := \omega t$. Observe that the required power $P$ is twice the angular frequency $\omega$ of the excitation torque, meaning that it is sufficient to consider the scaled time quantity in the interval $\psi \in [0,\pi]$ for the following analysis. Since required power is a continuous function of $\psi$, inverse image of the open interval $\Psi =  \{\psi : P(\psi) > P_{\max}\}$ is open \cite{munkres.BOOK2000}, i.e., $P^{-1}(\Psi)~\cong~(\psi_l, \psi_u)~\subset ~[0,\pi]$ with
\begin{equation}
P(\psi_l) = P(\psi_u) = P_{\max}.
\label{eq:openinterval}
\end{equation}
Solutions to lower and upper bounds in (\ref{eq:openinterval}) are obtained as 
\begin{equation*}
\!\psi_{\binom{l}{u}} = \dfrac{\pm \arccos(\cos\phi - 2 P_{\max} / (A^2 X))-\phi}{2} \!\!\mod\!\pi\!.
\end{equation*}
Therefore, when $\psi \in (\psi_l, \psi_u)$ (i.e., active region of power supply limit), power saturation maps desired motor torques to
\begin{equation*}
\bar{u} := {\rm psat}(u, \dot{q}) = P_{\max}/\left(A X \sin (\psi + \phi)\right).
\end{equation*}


Now that input-output torque relations of power supply limit are obtained in both cases, Fourier coefficients of the power limit nonlinearity over a period of the torque input $u~=~A~\sin(\omega t)$ i.e., $ t \in [0,2\pi]$, can be computed as
\begin{equation}
\begin{aligned}
c_N:=X_N \cos\phi_N &= \dfrac{1}{\pi A}\int_{0}^{2\pi} {\rm psat}(u,\dot{q}) \sin\psi \;{\rm d}\psi \\
s_N:=X_N \sin\phi_N &= \dfrac{1}{\pi A}\int_{0}^{2\pi} {\rm psat}(u,\dot{q}) \cos\psi \;{\rm d}\psi.
\end{aligned}
\label{eq:fourier}
\end{equation}
Defining the following notation for limits of the integral,
\begin{equation*}
\int_{l_1\cdots l_n}^{u_1\cdots u_n} h(x) {\mathrm d}x = \sum_{i=1}^{n} \int_{l_i}^{u_i} h(x) {\mathrm d}x,
\end{equation*}
integrals in (\ref{eq:fourier}) can be partitioned into active and passive regions of power saturation, corresponding to $P \leq P_{\max}$ and $P > P_{\max}$, respectively, as
\begin{equation*}
\begin{aligned}
\!\!\!c_N \!&= \!\!\dfrac{2}{\pi A}\!\left[\int\displaylimits_{0,\psi_u}^{\psi_l,\pi}\!\!\!\!A \sin^2\psi \;{\rm d}\psi +\! \int\displaylimits_{\psi_l}^{\psi_u}\!\!\bar{u}\sin\psi \;{\rm d}\psi \!\right] \\
\!\!\!s_N \!&= \!\!\dfrac{1}{\pi A}\!\left[\int\displaylimits_{0,\psi_u,\pi+\psi_u}^{\psi_l,\pi+\psi_l,2\pi}\!\!\!\!\!\!\!\!\!\!\!\tfrac{A}{2} \sin2\psi \;{\rm d}\psi +\!\!\!\!\!\!\!\! \int\displaylimits_{\psi_l,\pi+\psi_l}^{\psi_u,\pi+\psi_u}\!\!\!\!\!\!\!\!\!\bar{u} \cos\psi \;{\rm d}\psi \!\right].
\end{aligned}	
\end{equation*}
After some calculus to evaluate integrals, the following nonlinear system of equations are obtained
\begin{equation*}
\begin{aligned}
c_N &= \dfrac{2 P_{\max}}{\pi A^2 X}\;Y_N + \dfrac{2 (\pi-\Delta\psi) + \sin 2\psi_u - \sin 2\psi_l}{2\pi}\\
s_N &= \dfrac{2 P_{\max}}{\pi A^2 X}\;Z_N + \dfrac{\cos 2\psi_u - \cos 2\psi_l}{2\pi}
\end{aligned}
\end{equation*}
with $L_{\psi}~:=~\ln \dfrac{\sin(\psi_u + \phi)}{\sin(\psi_l + \phi)}$, $Y_N := \Delta\psi \cos\phi -L_{\psi} \sin\phi$, $Z_N:=\Delta\psi \sin\phi + L_{\psi} \cos\phi$, and $\Delta \psi~:=~\psi_u - \psi_l$. Since these equations do not admit analytical closed form solutions, numerical methods can be employed to obtain solution to $X_N$ and $\phi_N$ at a given amplitude $A$ and frequency $\omega$, which define the describing function $N(A,\omega)~=~X_N(A,\omega)~\phase{\phi_N(A,\omega)}$.

\begin{exmp}
	Consider a position servo system characterized by a second-order closed-loop transfer function 
	\begin{equation*}
	T(s) = \dfrac{\omega_n^2}{s^2 + 2 \,\zeta \omega_n s + \omega_n^2}
	\end{equation*}
	with natural frequency $\omega_n = 50\pi rad/s$ and damping ratio $\zeta=0.8$. Mechanics of the servo system consists of a mass $m~=~1~kg\,m^2$ and viscous damping $d~=~0.05~Ns/rad$ corresponding to a plant with transfer function 
	\begin{equation}
	G(s)=\dfrac{1}{s m + d}.
	\label{eq:universalexamplemodel}
	\end{equation} The system is controlled by a proportional-derivative (PD) controller in the form $u = K_p (q_d-q)-K_d\dot{q}$ with desired position $q_d$ and PD gains $K_p~=~\omega_n^2$ and $K_d~=~2 \zeta \omega_n~-~d$, respectively. Assume there exists a power limit $P_{\max}~=~400 W$. We consider both exact and approximate models of the power limit to evaluate their effects on the control loop through describing functions computed for amplitudes in the range $A \in [1, 500] rad$. Note that the approximate model is established as a standard torque saturation $u_{\max}~=~100 Nm$ which follows from the assumption that speed is limited by $\dot{q}_{\max} = 4 rad/s$.
	
	To begin with, the open-loop transfer function of the PD controlled mechanical system including a describing function is evaluated on a Nyquist plot in Fig. \ref{fig:nyquist}. It is observed that magnitudes of both nonlinearity models are less than or equal to unity, and the exact model produces non-negative imaginary part whereas approximation results in purely real describing function. These findings lead to two main conclusions: 1~)~Nonlinearities have no direct effect on stability since their describing functions do neither intersect with or encircled by the open-loop curve. 2) Nonlinearities increase the phase margin commonly, hence the robustness of the closed loop system. However, it should be distinguished that exact nonlinearity model provides more phase margin (even more with increase in amplitude) since it brings phase lead into the system because of its positive imaginary part as opposed to saturation. 
	
\end{exmp}

\begin{figure}[!h]
	\centering
	\includegraphics[width=0.8\columnwidth]{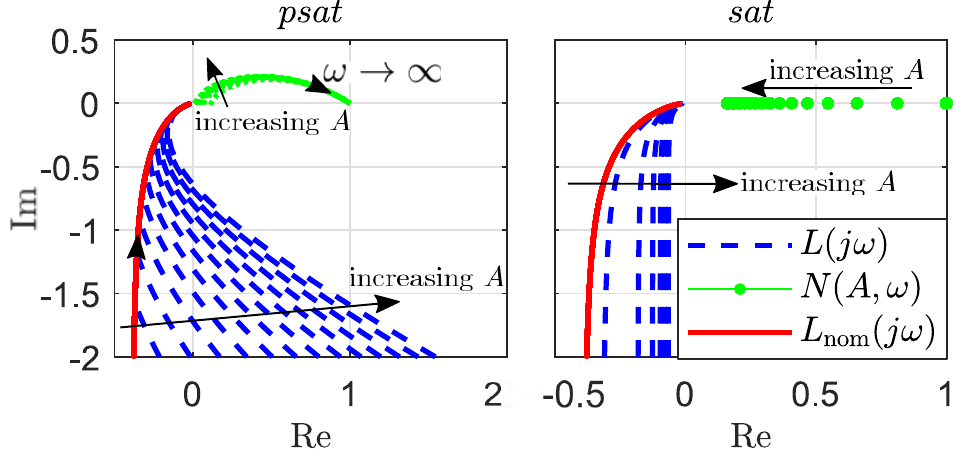}
	\caption{Nyquist diagram of nonlinearity describing function (dotted green), open-loop transfer function in the absence of power constraint (solid red), and open-loop transfer function when coupled with nonlinearity (dashed blue).}
	\label{fig:nyquist}
\end{figure}

\subsection{Maximum Closed Loop Bandwidth}
\label{sec:bandwidth}
In order to further improve on results of the last subsection summarized in Figure~\ref{fig:nyquist},
the closed loop bandwidth is evaluated more rigorously in this subsection by considering other nonlinearities common in motion control systems. While doing so, the exact model of power limit is quantitatively compared to its approximate model in terms of the achievable maximum bandwidth in the linear region of dynamics. Therefore, describing function of nonlinearity models are not needed. 

Consider a position controller with a sinusoidal reference $q_d = Y \cos (\omega t - \theta)$ applied to a single DoF linear system 
\begin{equation*}
m \ddot{q}  = u - k q - d \dot{q} - \tau_c \, {\rm sign}(\dot{q})
\end{equation*}
with inertia/mass $m$, viscous friction coefficient $d$, and Coulomb friction $\tau_c$. We are interested in maximum attainable bandwidth of the closed-loop system in the linear region of dynamics under different physical constraints. First, consider a power supply limit $P_{\max}$ (i.e., $P\leq P_{\max}$) and no-load speed limit $ -\dot{q}_{\max} \leq \dot{q} \leq \dot{q}_{\max}$. For the remaining cases, standard torque saturation $-u_{\max} \leq u \leq u_{max}$ with the conservative approximation to torque limit, i.e., $u_{\max} = P_{\max}/\dot{q}_{\max}$ given in (\ref{eq:psat_approx}), is modeled in addition to the no-load speed limit.

At cut-off frequency $\omega_c$, tracking the reference signal as
$q = \bar{Y} \cos(\omega_c t)$
with amplitude $\bar{Y} := Y /\sqrt{2}$ and phase delay $\theta$
requires velocity, torque and power
\vspace{-0.8cm}
{\small
\begin{equation*}
\begin{aligned}
\dot{q} &= -\bar{Y} \omega_c \sin(\omega_c t) \\
u &= \bar{Y} \left[(k - m \omega_c^2) \cos(\omega_c t) - d\, \omega_c \sin(\omega_c t) \right] - \tau_c {\rm sign}(\sin(\omega_c t)) \\
P &= -\tfrac{\bar{Y}^2 \omega_c}{2} \left[(k - m\omega_c^2) \sin (2\omega_c t) + d\,\omega_c (\cos(2\omega_c t) - 1) \right]\\ & \quad\, - \bar{Y} \omega_c \tau_c\lvert\sin(\omega_c t)\rvert.
\end{aligned}
\end{equation*}}In order to characterize the maximum closed loop bandwidth under different constraints, a common numerical optimization problem is formulated as follows : 
\begin{equation*}
\begin{aligned}
& \underset{\omega_c}{\text{maximize}}
& & \omega_c \\
& \text{subject to}
& & \omega_c > 0 &\cdots \quad(1)\\
& \quad 
& & \max_{0 \leq t \leq 2\pi/\omega_c} \lvert\dot{q}(t)\rvert \leq \dot{q}_{\max} &\cdots \quad(2)\\
& \quad
& & \max_{0\leq t \leq \pi/\omega_c} P(t) \leq P_{\max} &\cdots \quad(3)\\
& \quad
& & \max_{0\leq t \leq 2 \pi/\omega_c} \lvert u(t)\rvert \leq u_{\max} &\cdots \quad(4)
\end{aligned}
\end{equation*}
For the first case, which uses the exact model of power supply limit, inequality (4) is removed whereas inequality (3) is removed for the remaining case (i.e., approximate model).
\begin{exmp}
	Consider the same setup used in the Numerical Example of Sec.~\ref{sec:describingfunction} (i.e., rotational mechanical system, motor, and the driver). The theoretical permissible bandwidth of the system with exact and approximate models are evaluated for amplitudes in the range $Y \in (0, 1]$. This analysis is repeated for different power supply limits, specifically $P_{\max} \in \{200, 400, 600\} W$. Results are illustrated in Fig.~\ref{fig:maxclfreq} by a ratio of permissible bandwidths of the exact model to the approximate model, denoted by $\omega_c^{\rm psat}$ and $\omega_c^{\rm sat}$, respectively. Using exact model increases the maximum closed loop frequency substantially, nearly doubling the bandwidth at low amplitude region (e.g., $Y~<~10$ degrees for $P_{\max}~=~600 W$). Furthermore, it is seen from the figure that the ratio, interpreted as a function of amplitude $Y$, is not differentiable at two points. This can be explained by the change of active nonlinearity. In particular, the speed limit, which is passive for both models at small amplitudes, becomes active first for the exact model with the inactivation of power limit, and then for the approximate model with the inactivation of torque limit. For higher amplitudes, both systems are constrained by the speed limit alone, which makes the bandwidth ratio unity.
\end{exmp}

\begin{figure}[!h]
	\centering
	\includegraphics[width=0.9\columnwidth]{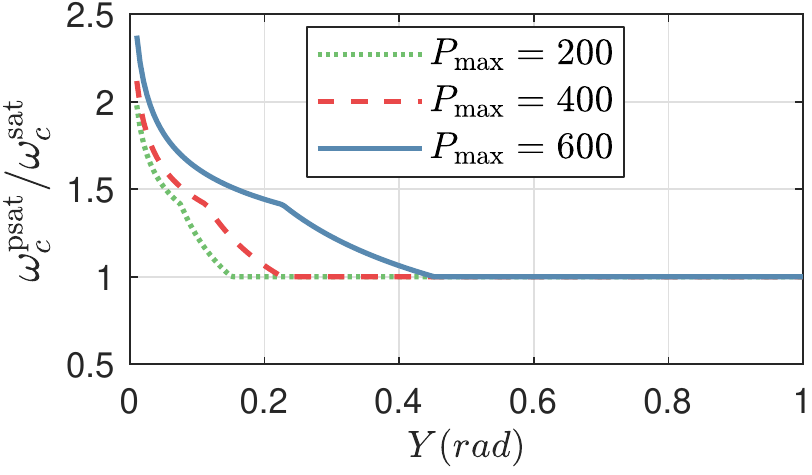}
	\caption{Maximum closed-loop frequency as a function of input amplitude for different springs.}
	\label{fig:maxclfreq}
\end{figure}

\section{Controller Design with Static Power Allocation}
\label{sec:controller}
Frequency domain analysis in Sec.~\ref{sec:frequencydomain} assesses performance and local stability of an actuator supplied by an energy source with a power supply limit in comparison to an actuator under a torque saturation corresponding to the conservative approximation (\ref{eq:psat_approx}) of the power limit. In particular, performance is measured by theoretical bandwidth limit of an actuator in the linear operation region, whereas describing function can serve as a tool to conduct local stability analysis around predicted limit cycles. In summary, the presented analysis suggests that using an exact model of the power supply limit outperforms the approximate model in terms of the open-loop phase margin, which is an indicator of closed-loop robustness, and closed-loop bandwidth. These results motivate further investigation of exact models for the power supply limit, especially during controller design. This section explicitly focuses on the ways of including power limit into existing motion control frameworks. In particular, we consider classical controllers for Euler-Lagrange mechanical systems. 

\subsection{Plant}
\label{sec:plant}
Linear and nonlinear mechanical systems are studied. Consider first, linear systems. Underactuated and fully-actuated linear mechanical systems can be commonly described by 
\begin{equation}
M \ddot{q} + D \dot{q} + K q = S \,{\rm psat}(u, \dot{q})
\label{eq:eom_linear}
\end{equation}
with mass matrix $M~=~M^T~>~0$, friction $D~=~D^T~\geq~0$, stiffness $K~=~K^T~\geq~0$, and actuator selection matrix $S$. Following, \cite{reyhanoglu.TAC1999}, the generalized coordinates can be decomposed into passive (unactuated) and actuated joints as $q = \left[q_u^T \; q_a^T\right]^T$ with unactuated joint vector $q_u \in \mathbb{R}^{n_u}$ and actuated joint vector $q_a \in \mathbb{R}^{n_a}$, respectively. Hence, the actuator selection matrix takes the form 
\begin{equation*}
S = \begin{bmatrix}
0_{n_u \times n_a} \\
I_{n_a \times n_a}
\end{bmatrix}.
\end{equation*}
On the other hand, fully-actuated rigid manipulators with nonlinear dynamics take the form of Euler Lagrange systems with Rayleigh dissipation \cite{ortega.BOOK1998}, which admit dynamics
\begin{equation}
M(q) \ddot{q} + C(q,\dot{q}) \dot{q} + D \dot{q} + G(q) = {\rm psat}(u, \dot{q})
\label{eq:eom_nonlinear}
\end{equation}
with mass matrix $M(q) = M^T(q) > 0$, Coriolis matrix $C(q,\dot{q})$, gravitational force $G(q)$ derived from potential field $\mathcal{U}(q)$ for all generalized coordinates and velocities in tangent bundle of configuration space $\mathcal{Q}$, i.e., $\forall (q^T, \dot{q}^T)^T \in T \mathcal{Q}$. Note that Coriolis matrix satisfies the skew-symmetric property \cite{mls.BOOK1994}
\begin{equation}
\dot{q}^T \left[\tfrac{1}{2} \dot{M}(q) - C(q,\dot{q})\right] \dot{q} = 0, \quad\quad \forall \; q, \dot{q} \in \mathbb{R}^n.
\label{eq:skewsymmetry}
\end{equation}
\subsection{Controls}
\label{sec:dynamiccontrol}
\subsubsection{Linear Systems}
For the plant with linear dynamics in (\ref{eq:eom_linear}), assume a linear dynamic controller of the form
\begin{equation}
\begin{aligned}
\dot{x}_c &= A_c x_c + B_p q + B_d \dot{q} \\
u &= C x_c + K_p q + K_d \dot{q}
\end{aligned}
\end{equation}
to stabilize the system in the absence of power limit saturation where $x_c \in \mathbb{R}^{n_c}$. Furthermore, the controller can be augmented by adding anti-windup compensation to avoid windup of controller states $x_c$ for which we consider both static and dynamic approaches. In particular, conditional integration (CI) and modern anti-windup (MAW) techniques \cite{zaccarian.EJC2009} with a general formulation are evaluated. These techniques are selected because of the following reasons : 1) CI is commonly used in practice since it is simple in the sense that it does not introduce any additional parameters for anti-windup augmentation \cite{visioli.BOOK2006}. 2) MAW provides superior performance compared to static techniques by guaranteeing stability and optimal performance with the additional flexibility encoded by the dynamics of anti-windup augmentation \cite{zaccarian_teel.BOOK2011}.

We propose to update the controller states according to the CI rule
\begin{equation}
\dot{x}_c = \mathcal{N}[\rho(C,\mathcal{H})] \, (A_c x_c + B_p q + B_d \dot{q})
\end{equation}
where $\mathcal{H}=\{1 \leq i \leq n_a : P_i > \bar{P}_i\}$ denotes the set of actuator channels demanding more power than the limit, $\rho(C,\mathcal{H})$ selects the rows of $C$ belonging to the set $\mathcal{H}$  
\begin{equation*}
\rho(C,\mathcal{H}) = [c_i], \forall i \in \mathcal{H}
\end{equation*}
with $\rho(C,\emptyset) = I$, and $\mathcal{N}[\rho(C,\mathcal{H})]$ denotes the projection operator onto null-space of the matrix $\rho(C,\mathcal{H})$ formulated as
\begin{equation}
\mathcal{N}[\rho(C,\mathcal{H})] = V_\mathcal{H} \Sigma_\mathcal{H} (V_\mathcal{H})^T 
\end{equation}
with singular value decomposition $\rho(C,\mathcal{H}) = U_\mathcal{H} \Sigma_\mathcal{H} V_\mathcal{H}^T$. The null-space projection freezes the integrator of state vector in the subspace associated with saturating actuator channels. Under the controller action, the closed loop system can be written as
\begin{equation}
\dot{x} = A(\mathcal{H}) x + B \, \sigma(x)
\label{eq:closedloop}
\end{equation}
where 
\begin{equation*}
\begin{aligned}
x &:= \!\! 
\begin{bmatrix}
q \\
\dot{q} 
\\
x_c
\end{bmatrix}, \quad
\sigma(x) := {\rm psat}(u, \dot{q}) - u, \quad u=\kappa x, \\
A(\mathcal{H}) &:= \!\!\begin{bmatrix}
0 & I &  0\\
M^{-1}(K_p + K) & M^{-1}(K_d + D) & M^{-1} C \\
\mathcal{N}[\rho(C,\mathcal{H})] B_p & \mathcal{N}[\rho(C,\mathcal{H})] B_d & 	
\mathcal{N}[\rho(C,\mathcal{H})] A_c
\end{bmatrix}\!, \\
B &:=\!\!
\begin{bmatrix}
0 \\
M^{-1} S \\
0
\end{bmatrix}, \text{ and }	
\kappa := \begin{bmatrix}
K_p & K_d &  C
\end{bmatrix}.	
\end{aligned}
\end{equation*}
As a second anti-windup alternative, MAW scheme can be used in which controller dynamics and controller output take the form 
\begin{equation}
\begin{aligned}
\dot{x}_c &= A_c x_c + B_p q + B_d \dot{q} + E_c \sigma(x) \\
u &= \kappa x + E \sigma(x)
\end{aligned}
\end{equation}
where $x_c$ now includes anti-windup augmentation states, and $E_c$ and $E$ denote anti-windup compensation gains contributing to the controller dynamics and output, respectively.
Thus, the closed loop system dynamics in (\ref{eq:closedloop}) is obtained with
\begin{equation*}
\begin{aligned}
A(\mathcal{H}) &:= \!\!\begin{bmatrix}
0 & I &  0\\
M^{-1}(K_p + K) & M^{-1}(K_d + D) & M^{-1} C \\
B_p & B_d & A_c
\end{bmatrix} \text{ and}\\
B &:=\!\!
\begin{bmatrix}
0 \\
M^{-1} S (I + E) \\
E_c
\end{bmatrix}.
\end{aligned}
\end{equation*}For the function $\sigma(x)$ which is common in both controller structures, a closed convex polytope can be defined as
\begin{equation*}
\mathcal{L} = \{x \in T \mathcal{Q} \times \mathbb{R}^{n_c} : \lvert\gamma_i \, [K_c]_i x\rvert \leq \bar{P}_i/\bar{v}_i , i = 1, \dots, n_a \}
\end{equation*}
with $[K_c]_i = \rho(K_c, \{i\})$ and a positive scalar $\gamma_i < 1$, such that $\forall x \in \mathcal{L}$, the following holds
\begin{equation*}
\begin{aligned}
0 \leq	\gamma_i [K_c]_i x \leq {\rm psat}([K_c]_i x, \dot{q}) \leq [K_c]_i x & \quad \text{if } [K_c]_i x \geq 0 \\
[K_c]_i x \leq {\rm psat}([K_c]_i x, \dot{q}) \leq \gamma_i [K_c]_i x \leq 0 & \quad \text{otherwise}.	
\end{aligned}
\end{equation*}
This translates to
\begin{equation}
\begin{aligned}
- (1-\gamma_i) [K_c]_i x \leq [\sigma(x)]_i \leq 0 & \quad \text{if } [K_c]_i x \geq 0 \\
0 \leq [\sigma(x)]_i \leq -(1-\gamma_i) [K_c]_i x & \quad \text{otherwise}.	
\end{aligned}
\label{eq:sigmavertices}
\end{equation}
\begin{thm}
	The system in (\ref{eq:closedloop}) is asymptotically stable in  $\mathcal{L}$ if 
	\begin{equation*}
	x^T \left[\left(A(\mathcal{H}) + B \,\Pi(\mathcal{H})\right)^T Q + Q \left(A(\mathcal{H}) + B \,\Pi(\mathcal{H})\right) \right] x < 0
	\end{equation*}
	for all $\mathcal{H} \in \euscr{P}(n_a)$ with a matrix $Q~=~Q^T~>~0$, $\euscr{P}(n_a)$ denoting the power set of integers from 1 to $n_a$, which, having cardinality of $2^{n_a}$, $\delta(\mathcal{H})$ a mapping from indices of saturating joints to $\mathbb{R}^{n \times n_c}$ defined as
	\begin{equation*}
	\Pi(\mathcal{H}) = \begin{bmatrix}
	- \delta(1,\mathcal{H}) (1-\gamma_1) [K_c]_1 \\ 
	- \delta(2,\mathcal{H}) (1-\gamma_2) [K_c]_2 \\
	\vdots \\
	- \delta(n_a,\mathcal{H}) (1-\gamma_{n_a}) [K_c]_{n_a}
	\end{bmatrix},
	\end{equation*}	
	and a binary indicator function
	\begin{equation*}
	\delta(i,\mathcal{H}) = \begin{cases}
	1 & \text{ if } i\in \mathcal{H} \\
	0 & \text{ otherwise.}
	\end{cases}
	\end{equation*}
\end{thm}
\begin{proof}
	Consider a Lyapunov function $\mathcal{V}(x) = x^T Q x$. The set $\mathcal{L}$ is asymptotically stable if 
	\begin{equation}
	\!\dot{\mathcal{V}}(x)\!=\!x^T \left[A(\mathcal{H})^T Q + Q A(\mathcal{H})\right] x + 2 x^T Q B \sigma(x)\!<\!0,
	\label{eq:linearstability}
	\end{equation}
	for all $x \in \mathcal{L}$. In other words, it can be said that the set $\mathcal{L}$ is inside the region of attraction (ROA). Since $\dot{\mathcal{V}}(x)$ is linear in $\sigma(x)$ and $A$, it attains its maximum in $\mathcal{L}$ at one of the vertices defined in (\ref{eq:sigmavertices}), which implies 
	\begin{equation*}
	\dot{\mathcal{V}}(x) \leq \max_{\mathcal{H} \in \euscr{P}} x^T \left[A(\mathcal{H})^T Q + Q A(\mathcal{H})\right] x + 2 x^T Q B \,\Pi(\mathcal{H}) x.
	\end{equation*}
	Hence, (\ref{eq:linearstability}) is satisfied if, for all $\mathcal{H} \in \euscr{P}$, we have 
	\begin{equation*}
	x^T \left[\left(A(\mathcal{H}) + B \,\Pi(\mathcal{H})\right)^T Q + Q \left(A(\mathcal{H}) + B \,\Pi(\mathcal{H})\right) \right] x < 0.
	\end{equation*}	
\end{proof}
\vspace{-3mm}
The ROA can be estimated by fitting a prescribed shape to the interior of polytope $\mathcal{L}$ \cite{tarbouriech.CDC1997} and \cite{lin.AUTO2002}. In this paper, ellipsoid $\mathcal{E} (Q, \alpha) = \{x : x^T Q x \leq \alpha \} \subset \mathcal{L}$, identified by the Lyapunov function and $\alpha~>~0$, is chosen as the reference shape. Therefore, volume of ROA can be maximized by jointly tuning the scalar $\gamma$, which determines the affine inequalities of polytope $\mathcal{L}$, and $\alpha$, which determines the boundary of ellipsoid, in line with prior works \cite{boyd.CDC1998} and \cite{lin.AUTO2002}. Using Schur's complement, this can be translated to a constrained optimization problem defined in the language of LMIs as
\begin{equation}
\begin{aligned}
& \underset{\substack{A_c, B_p, B_d \\ C, K_p, K_d \\ W, [\gamma_i]}}{\text{max}}
& & \log\det W \\
& \text{s.t.}
&& W > 0, \\
& \quad
& & 0 < \gamma_i < 1 \; \forall i = 1 \dots n_a, \\
& \quad
& & W \left(A(\mathcal{H}) + B \,\Pi(\mathcal{H})\right)^T + \left(A(\mathcal{H}) + B \,\Pi(\mathcal{H})\right) W < 0 \\ 
& \quad && \forall \mathcal{H} \in \euscr{P}(\mathcal{H}), \\
& \quad
& & \gamma_i^2 h_i W h_i^T  \leq 1 \Leftrightarrow \begin{bmatrix}
1 & \gamma_i h_i W \\
\gamma_i W h_i^T & W
\end{bmatrix} \geq 0 \\
& \quad && \forall i = 1 \dots n_a \\
\end{aligned} 
\end{equation}
with $W := \left(\dfrac{Q}{\alpha}\right)^{-1}$ and $h_i := \dfrac{\bar{v}_i}{\bigstrut \bar{P}_i} [K_c]_i$. Unfortunately, this program is not convex for joint optimization of controller gains and remaining parameters, which are $W$ and $\gamma_i$s, to find the controller with maximal ROA since they are not simultaneously linear in the LMIs. Iterative methods can be used to decompose the problem into convex subproblems and solve each stage of nested optimization via convex methods \cite{lin.AUTO2002, lin.AUTO2004}. Nevertheless, in Sec.~\ref{sec:experiments}, a controller with a CI anti-windup augmentation based on the exact model of power supply limit is designed following the procedure developed so far in this section and evaluated experimentally in both time and frequency domains compared to the same controller using the approximate model (\ref{eq:psat_simple}) of the power limit. 

\subsubsection{Nonlinear Systems} 
\label{sec:pbc}
Consider now a fully-actuated nonlinear Euler-Lagrange system controlled by a passivity-based static nonlinear controller (also known as PD plus gravity compensation) in the form
\begin{equation}
u = G(q) - K_p q - K_d \dot{q}.
\label{eq:pbc}
\end{equation}
with $K_p=K_p^T>0$. A Lyapunov function candidate
\begin{equation}
\mathcal{V}(q,\dot{q}) = \tfrac{1}{2} \dot{q}^T M(q) \dot{q} + \tfrac{1}{2} q^T K_p q,
\label{eq:pbclyapunov}
\end{equation}
can be derived from the kinetic energy along with the "shaped" potential energy \cite{koditschek.TECHREP1989} \cite{tomei.TRO1991}. This Lyapunov function candidate is radially unbounded and positive definite. Its derivative along the trajectories of (\ref{eq:eom_nonlinear}) can be found as follows : First differentiate it with respect to time
\begin{equation}
\dot{\mathcal{V}}(q,\dot{q}) = \dot{q}^T M(q) \ddot{q} + \dfrac{1}{2} \dot{q}^T \dot{M}(q) \dot{q} + \dot{q}^T K_p q.
\end{equation}
Then, substituting system dynamics in (\ref{eq:eom_nonlinear}) gives
\begin{equation}
\begin{aligned}
\dot{\mathcal{V}}(q,\dot{q}) =& \dot{q}^T \left({\rm psat}(u,\dot{q}) - C(q,\dot{q}) \dot{q} - D \dot{q} - G(q)\right) \\ +& \dfrac{1}{2} \dot{q}^T \dot{M}(q) \dot{q} + \dot{q}^T K_p q
\end{aligned}
\end{equation}
which can be further simplified to
\begin{equation}
\begin{aligned}
\dot{\mathcal{V}}(q,\dot{q}) =& \dot{q}^T {\rm psat}(u,\dot{q}) + \dfrac{1}{2} \dot{q}^T \left(\dot{M}(q) - 2C(q,\dot{q})\right)\dot{q} \\ 
+& \dot{q}^T K_p q - \dot{q}^T G(q) - \dot{q}^T D \dot{q}.
\end{aligned}
\end{equation}
Using the skew-symmetric property in (\ref{eq:skewsymmetry}) and substituting the control action in (\ref{eq:pbc}) after adding/subtracting the term $\dot{q}^T K_d \dot{q}$, we obtain
\begin{equation}
\dot{\mathcal{V}}(q,\dot{q}) = \dot{q}^T {\rm psat}(u,\dot{q}) - \dot{q}^T u - \dot{q}^T (K_d + D) \dot{q}.
\end{equation}
Consider, first, $P_i \leq \bar{P}_{i}$ for all joints meaning that ${\rm psat}(u,\dot{q}) = u$. This yields
\begin{equation}
\dot{\mathcal{V}}(q,\dot{q}) = - \dot{q}^T (K_d + D) \dot{q} \leq 0
\label{eq:pbclyapunovrate}
\end{equation}
negative as long as $\dot{q}\neq0$ for all system trajectories. With ${\rm psat}(u,\dot{q}) = u$, the system dynamics reduce to 
\begin{equation}
M(q)\ddot{q} + C(q,\dot{q}) \dot{q} + K_p q + (K_d + D) \dot{q} = 0
\end{equation}
which renders the origin ($q =0$ and $\dot{q}$) unique equilibrium point in the state space. Then, application of LaSalle's Invariance Principle \cite{khalil.BOOK2002} ensures that the origin is globally asymptotically stable (GAS). For the remaining case with $P_i > \bar{P}_i$, the time-derivative of the candidate Lyapunov function becomes strictly negative for all $q$ and $\dot{q}$ with 
\begin{equation}
\dot{\mathcal{V}}(q,\dot{q}) = - \dot{q}^T (K_d + D) \dot{q} - \sum_{i \in F}P_i - \bar{P}_i < 0,
\end{equation}
thus concluding that the system under passivity-based control action subjected to power supply limits is GAS.

\begin{table}[!h]
	\centering
	\caption{Physical parameters of the two-link robot model.}
	\label{tab:pbcsim}
	  \setlength\extrarowheight{-3pt}
	\begin{tabular}{r||c|c}
		{\bf Parameter} & {\bf Symbol} & {\bf Value} \\ \hline
		Link 1 mass & $m_1$ & 16 kg \\
		Link 2 mass & $m_2$ & 12 kg \\
		Link 1 inertia & $I_1$ & 18 $\text{kg}\,\text{m}^2$\\
		Link 2 inertia & $I_2$ & 7.5 $\text{kg}\,\text{m}^2$\\
		Joint damping matrix & $D$ & diag(10,10) $\text{N}\text{m}\,\text{s}/\text{rad}$ \\
		Link 1 length & $h_1$ & 1 m \\
		Link 2 length & $h_2$ & 1 m \\
		Link 1 angle & $q_1$ & $(-)$ rad \\
		Link 2 angle & $q_2$ & $(-)$ rad \\
		Link 1 torque & $u_1$ & $(-)$ $\text{N}\,\text{m}$ \\
		Link 2 torque & $u_2$ & $(-)$ $\text{N}\,\text{m}$ \\
		Gravity & $g$ & 9.8 $\text{m}/\text{s}^2$ \\
	\end{tabular}
\end{table}	

\begin{figure}[!b]
	\centering
	\includegraphics[width=0.6\columnwidth]{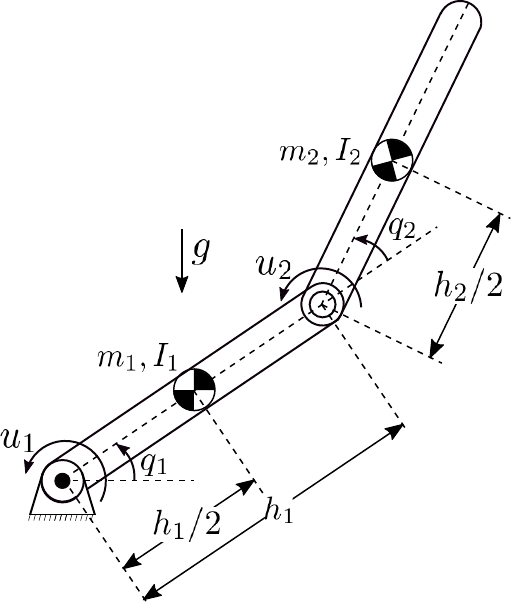}
	\caption{Fully-actuated two-link planar manipulator.}
	\label{fig:twolink}
\end{figure}

\begin{exmp}
	Consider a manipulator illustrated in Fig.~\ref{fig:twolink}. Assume that joint actuators are supplied by a single power source with peak power $P_{\max} = 2 kW$, corresponding to power limit of each joint $\bar{P}_1 = \bar{P}_2 = P_{\max} / 2$. Using parameter values given in Table~\ref{tab:pbcsim} and a passivity-based controller 
	\begin{equation*}
	\begin{aligned}
	K_p &= diag\left(\left[m_1 \omega_n^2, \, m_2 \omega_n^2\right]\right) \\
	K_d &= diag\left(\left[2 m_1 \zeta \omega_n, \, 2 m_2 \zeta \omega_n \right]\right),
	\end{aligned}
	\end{equation*}
	with $\omega_n = 2 \pi \sqrt{2}$ and $\zeta = 0.9$, the robot is simulated from initial conditions $q_1=-\pi/2$ and $q_2 = \pi$. As illustrated in Fig.~\ref{fig:pbcsim}, joint trajectories and joint torques verify that the origin is asymptotically stable for this system even though mechanical power of each joint actuator is upper bounded by the power supply limit of $1\,kW$.
	\begin{figure}[!h]
		\centering
		\includegraphics[width=0.7\columnwidth]{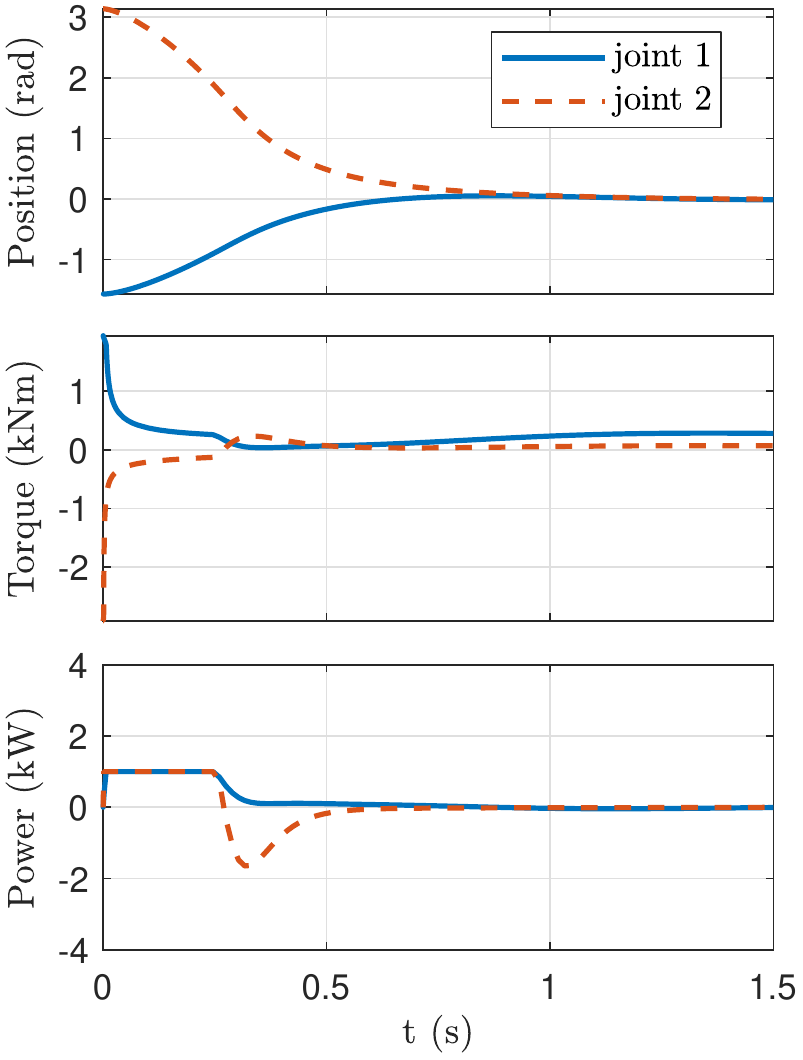}
		\caption{Trajectories (top), torques (middle), and power consumption (bottom) of a 2-DoF rotary robot with a passivity-based controller.}
		\label{fig:pbcsim}
	\end{figure}
\end{exmp}

\section{A Unified Algorithm for Dynamic Allocation and Control}
\label{sec:mpc}
Classical controllers with static power allocation enforce $P_j \leq \bar{P}_j$ for all joints. However, this may be unnecessary in some cases because channelling more power to particular joints (i.e. $P_j~>~\bar{P}_j$) without compromising $\sum_{i=1}^{n_a} P_i < P_{\max}$ is sometimes a feasible strategy which can improve tracking performance. For example, when a joint requires power less than its predefined limit for a desired trajectory, its residual power may be used temporarily by another joint which uses its available power completely but still fails to follow its own desired trajectory. This motivates the use of dynamic power allocation for improved performance. In the following subsections, two quadratic programming (QP) problems are presented to simultaneously solve control and dynamic power allocation problems for linear and nonlinear systems formulated as a finite-horizon optimal control problem and a control Lyapunov function-based QP problem, respectively. 

\subsection{Linear Systems}
\label{sec:linearmpc}
Without loss of generality, we consider linear mechanical systems in (\ref{eq:eom_linear}) extended with constant gravitational disturbance $G$. This results in state-space model
\begin{equation}
\dot{x}	= F_c x + H_c \upsilon + g_c
\label{eq:optimalcontrol_statespace}
\end{equation}
\begin{flalign*}
&\text{with }
 F_c\!=\!\begin{bmatrix}
	0 & I \\
	-M^{-1} K\! & \!-M^{-1} D
	\end{bmatrix}\!\!,\,
 H_c\!=\!\begin{bmatrix}
	0 \\
	M^{-1} S
	\end{bmatrix}\!\!, \,x\!=\!\begin{bmatrix}
	q \\
	\dot{q}
	\end{bmatrix}\!\!, \\
&\text{ and } g_c = \begin{bmatrix}
	0 \\
	-M^{-1} G
	\end{bmatrix}. \text{ Zero-order-hold discretization of}&&
\end{flalign*}
dynamics with sample time $\Delta t$ yields
\begin{equation}
x_{k+1} = F x_k + H  \upsilon_k + g
\label{eq:mpc_dynamics}
\end{equation}
\begin{equation*}
\begin{aligned}
&\text{with } H\!\!=\!\!\left(\!\int\displaylimits_{0}^{\Delta t}\!\!e^{F_c t} {\rm d}t\!\!\right)\!\!H_c, \; g\!\!=\!\!\left(\!\int\displaylimits_{0}^{\Delta t}\!\!e^{F_c t} {\rm d}t\!\!\right)\!\!g_c, \text{ and }F\!=\!e^{F_c \Delta t}\!.
\end{aligned}
\end{equation*}

Given an initial condition $x_0$ and a sequence of control action $\{\upsilon_i\}$, the solution to discrete dynamics is obtained as 
\begin{equation}
x_n = F^n x_0 + \sum_{i = 0}^{n - 1} F^{n-1-i} (H \upsilon_i + g).
\label{eq:mpc_dynamicssolution}
\end{equation}

We are interested in finite-horizon optimal control of dynamical system (\ref{eq:mpc_dynamics}) in the presence of linear constraints on state and control inputs and power supply limit. A common choice of quadratic cost function would be 
\begin{equation}
J~=~\underbrace{\Delta x_N^T \Lambda_f \Delta x_N}_{\substack {terminal\\cost}}~+~\sum_{n=0}^{N-1} \underbrace{\Delta x_n^T \Lambda \Delta x_n~+~\Delta \upsilon_n^T \Phi \Delta \upsilon_n}_{\substack{running\\cost}}
\end{equation}
with $\Delta x_n := x_n - x^\star_n$ state error between actual states and reference state target $x^\star_n$, $\Delta \upsilon_n := \upsilon_n - \upsilon^\star_n$ control input error between actual control signal and  control signal target $\upsilon^\star_n$ corresponding to reference state $x^\star_n$, horizon length $N$, and symmetric weight matrices $\Lambda_f~=~\Lambda_f^T \geq 0$, $\Lambda~=~\Lambda^T \geq  0$, and $\Phi = \Phi^T > 0$. Substituting (\ref{eq:mpc_dynamicssolution}) yields an equivalent form of cost function as a quadratic function of torques
\begin{equation}
J(\Upsilon) = c_z + z \Upsilon_N + \Upsilon_N^T \,Z \,\Upsilon_N
\label{eq:cost}
\end{equation}
with 
\begin{equation*}
\begin{aligned}
c_z &:= (\bar{x}_0 + \bar{g} - X_N^\star)^T \tilde{\Lambda} (\bar{x}_0 + \bar{g} - X_N^\star) + \Upsilon_N^{\star^T} \tilde{\Phi} \Upsilon_N^{\star}, \\
z &:= 2 (\bar{x}_0 +\bar{g} - X_N^\star)^T \tilde{\Lambda} \hat{H} - 2 \Upsilon_N^{\star^T} \tilde{\Phi}, \quad \;Z:=\hat{H}^T \tilde{\Lambda} \hat{H} + \tilde{\Phi},\!\\
\tilde{\Lambda} &:= {\rm diag}(\underbrace{\Lambda, \cdots, \Lambda}_{\substack{N-1\\copies}}, \Lambda_f), \quad
\tilde{\Phi} := {\rm diag}(\underbrace{\Phi, \cdots, \Phi}_{\substack{N\\copies}}), \\
\Upsilon_N &:= \!\! 
\begin{bmatrix}
\upsilon_0 \\ 
\vdots \\
\upsilon_{N-1}
\end{bmatrix}, \;
\bar{x}_0 :=\!\!  \begin{bmatrix}
F^1 \\ 
\vdots \\[3pt]
F^N
\end{bmatrix}\! x_0,\;
\bar{g} := \! \hat{F}\!  \begin{bmatrix}
g \\ 
\vdots \\[3pt]
g 
\end{bmatrix}, \; 
X_N^\star := \!\! 
\begin{bmatrix}
x_1^\star \\ 
\vdots \\
x_N^\star
\end{bmatrix}, \;
\\
\Upsilon_N^\star &:= \!\! 
\begin{bmatrix}
\upsilon_0^\star \\ 
\vdots \\
\upsilon_{N-1}^\star
\end{bmatrix}\!, \;
\hat{H}:=\hat{F}\! \underbrace{\begin{bmatrix}
	H & H & \cdots & H\\
	0 & H & \cdots & H \\
	\vdots & \vdots & \ddots & \vdots\\
	0 & 0 & \cdots & H
	\end{bmatrix}}_{\substack{N \times N \\copies}}, \text{and } \\
\hat{F} &:=\! \begin{bmatrix*}[c]
I & 0 & \cdots & 0\\ 
F & I & \ddots & \vdots \\
\vdots & & \ddots & 0\\[4pt]
F^{N-1} & F^{N-2} & \cdots & I
\end{bmatrix*}\!.
\end{aligned}
\end{equation*}
Similarly, one can see from the direct application of (\ref{eq:mpc_dynamicssolution}) that linear constraints on states in the form
\begin{equation*}
a_{\rm neq} \, x_k \leq b_{\rm neq} \quad \forall k = 1,\hdots,N
\end{equation*}
with $a_{\rm neq} \in \mathbb{R}^{n_c \times 2(n_a+n_u)}$, $b_{\rm neq} \in \mathbb{R}^{n_c}$, and $\leq$ denoting the component-wise relation translates to a linear constraint on torques as
\begin{equation*}
A_{\rm neq} \Upsilon_N \leq B_{\rm neq}
\end{equation*}
with 
\begin{equation*}
\begin{aligned}
A_{\rm neq} &:= D_{\rm neq} \hat{H}, \;
B_{\rm neq}:=\bar{b}_{\rm neq} - D_{\rm neq} (\bar{g} + \bar{x}_0), \\
\bar{b}_{\rm neq} &:= \left[
b_{\rm neq}^T \cdots b_{\rm neq}^T\right]^T, \text{ and } D_{\rm neq} := {\rm diag}(\underbrace{a_{\rm neq}, \cdots, a_{\rm neq}}_{\substack{N-1\;copies}}).
\end{aligned}
\end{equation*}
Therefore, the optimal control problem can be formulated as
\begin{equation}
\begin{aligned}
& \underset{\Upsilon_N}{\text{min}}
& & J(\Upsilon) = c_z + z \Upsilon_N + \Upsilon_N^T \,Z \,\Upsilon_N\\
& \text{s.t.}
& & A_{\rm neq} \Upsilon_N \leq B_{\rm neq} \\
& \quad
& & (\dot{q}_a)_n^T \upsilon_n + \upsilon_n^T\,{\rm diag}(\bar{R})\,\upsilon_n \leq P_{\max}, \; 0\leq n \leq N-1
\end{aligned}
\label{eq:mpc_optimproblem}
\end{equation}
where $(\dot{q}_a)_n := S^T \dot{q}_n$ denotes the velocities of actuated joints at time step $n$, and $\bar{R} : = \left[\dfrac{R_1}{k_t^2}, \dots, \dfrac{R_{n_a}}{k_{t}^2}\right]^T$.
\begin{prop}
	Power supply limit defines a quadratic constraint on torques $\Upsilon$.
\end{prop}
\begin{proof}
	Observe that the solution to velocities takes the form
	\begin{equation*}
	\dot{q}_n  = S_q^T\left(F^n \; x_0 + \sum_{i=0}^{n-1} F^{n-1-i} \; (H \upsilon_i + g)\right),
	\end{equation*}
	where the operator $S_q^T = \left[0, I\right]^T$ selects the lower half of the states (i.e., velocities). Plugging this into (\ref{eq:mpc_optimproblem}) gives a quadratic inequality
	\begin{equation}
	\Upsilon_n^T 
	\underbrace{\left[\begin{array}{c c l : c}
		0 & \dots & 0 & C_0 \\
		\vdots & & \vdots & \vdots \\
		0 & \dots & 0 & C_{n-2} \\ \hdashline
		C_0^T & \dots & C_{n-2}^T & \Omega
		\end{array}	\right]}_{\displaystyle E}
	\Upsilon_n
	+ 
	\underbrace{\begin{bmatrix}
		0 \\
		\vdots \\
		0 \\ 
		\hdashline
		\beta
		\end{bmatrix}^T}_{\displaystyle \chi^T}\!\!\!\Upsilon_n \leq P_{\max}
	\label{eq:mpc_powerlimineq}
	\end{equation}
	with $C_i = \tfrac{1}{2}H^T (F^{n-1-i})^T S_q S$, $\Omega ={\rm diag}(\bar{R})$, and
	\begin{equation*}
	\beta = S^T S_q^T \left(F^n \; x_0 + \sum_{i=0}^{n-1} F^{n-1-i} \; g\right).
	\end{equation*}
	The inequality (\ref{eq:mpc_powerlimineq}) can be lifted to space of full control action sequence $\Upsilon_N$ by padding the matrix $E$ and the vector $\chi$ acting on $\Upsilon_n$ with zeros in the expression (\ref{eq:mpc_powerlimineq}). Hence, it can be concluded that the power limit imposes a quadratic constraint on $\Upsilon_N$.
\end{proof}

\begin{cor}
	The optimal control problem in (\ref{eq:mpc_optimproblem}) is a quadratically constrained quadratic program (QCQP) \cite{boyd.BOOK2004}.
\end{cor}

\begin{thm}
	QCQP in (\ref{eq:mpc_optimproblem}) is a non-convex problem.
\end{thm}
\begin{proof}
	
	First note that the cost function is convex since $Z~>~0$. 
	Further, linear constraints preserve convexity. Therefore, we restrict the proof to showing that power constraint is not convex. The inequality (\ref{eq:mpc_powerlimineq}) can be identified with the following quadratic form : 
	\begin{equation}
	\Upsilon_n^T\,E\,\Upsilon_n + \chi^T\,\Upsilon_n \leq P_{\max}
	\end{equation}
	This constraint is not convex if and only if $E$ is not positive semidefinite. Rest of the proof is done separately for the following two cases : i) $\bar{R} \neq 0$ ii) $\bar{R} =0$, corresponding to full model of power limit including electrical losses and reduced model without electrical losses, respectively.
	
	
	\begin{enumerate}[i)]
		\item Suppose that the theorem holds. This is equivalent to $E$ being positive semidefinite, from which Schur Complement defines the following conditions : 
		\begin{equation}
		\begin{aligned}
		E \geq 0 \Leftrightarrow \Omega \geq 0 \text{ and }
		\begin{bmatrix}
			C_0 \\
			\vdots\\
			C_{n-1}
			\end{bmatrix} \Omega
			\begin{bmatrix}
			C_0 \\
			\vdots\\
			C_{n-1}
			\end{bmatrix}^T \leq 0.
		\end{aligned}
		\label{eq:schurcomplement}
		\end{equation}
		It is obvious that the first inequality $\Omega \geq 0$ on the right hand side automatically refutes the inequality next to it which leads to a contradiction with the proposition $E \geq 0$. As a result, it is concluded that the power supply constraint is non-convex when $\bar{R} \neq 0$. 
		
		\item For a dynamical system (\ref{eq:mpc_dynamics}), the matrix $E$ can be treated as a continuous function of $\bar{R}$ since its entries depend continuously on $R_i$s. Given this fact, it follows from Theorem 3.1.1 in \cite{ortega.BOOK1990} that all eigenvalues of $E$ are continuous functions of $\bar{R}$, i.e., $\lambda_i(E(\bar{R})) \in \mathcal{C}^0$ for all $i \in {1, ..., n_a}$ where $\lambda_i(E)$ denotes the $i^{th}$ eigenvalue of $E$ and $\mathcal{C}^0$ space of continuous functions. Furthermore, it is known from the previous case that $E$ is not positive semidefinite which implies
		\begin{equation}
		\underset{i}{\min} \; \lambda_i(E(\bar{R})) < 0
		\label{eq:mineig}
		\end{equation}
		for $\bar{R} \neq 0$. Therefore, in order to show that the power supply constraint is non-convex still in this case, it is sufficient to verify (\ref{eq:mineig}) when $\bar{R} = 0$. Now, suppose that the contrary is true, i.e.,
		\begin{equation}
		\underset{i}{\min} \; \lambda_i(E(0)) \geq 0.
		\label{eq:mineig_nonneg}
		\end{equation}
		On the other hand, it follows from the fact 
		\begin{equation*}
		\rm{tr}(E(\bar{R})) = \displaystyle\sum_{i=1}^{n\times n_a} \lambda_i(E(\bar{R}))
		\end{equation*}
		that $\sum_{i}^{} \lambda_i (E(0)) = 0$. Equation (\ref{eq:mineig_nonneg}) and this relation can be satisfied simultaneously if and only if all eigenvalues are zero, which requires $E(0)$ to be a nilpotent matrix, which is obviously not true, leading to a contradiction with (\ref{eq:mineig_nonneg}).
	\end{enumerate}
\end{proof}

\begin{exmp}
	In order to validate the algorithm proposed in this section and demonstrate its efficiency, a fin actuation system of a missile illustrated in Fig.~\ref{fig:linearoptimalsim} is considered as an example system. The system consists of four identical actuators each connected to a fin used to control the missile by redirecting aerodynamic forces on it. Therefore, an actuator of the fin $i\in \mathcal{N}_a=\{1,2,3,4\}$ controls its position $q_i$ by applying torques $u_i$ which are subject to the following constraints : 
	\begin{enumerate}
		\item Actuators have fixed force limits (i.e., $-\bar{u} \leq u_i \leq \bar{u}$ associated with the peak current rating of their drive electronics.
		\item A realistic torque speed curve is enforced on actuators as a speed dependent force limit 
		\begin{equation*}
		- u_{{\rm stall}} \left(1 + \dfrac{\dot{q}_i}{\dot{q}_{{\rm max}}}\right) \leq u_i \leq u_{{\rm stall}} \left(1 - \dfrac{\dot{q}_i}{\dot{q}_{{\rm max}}}\right) \; 
		\end{equation*}
		where $u_{{\rm stall}}$ and $\dot{q}_{{\rm max}}$ denote stall force and no-load speed of actuators, respectively.
		\item The system is supplied by a common power source with maximum power output $P_{\max}$ which defines a constraint on torques depending on the power supply limit model used by the controllers.
	\end{enumerate} 
	\begin{figure}[!b]
		\centering
		\includegraphics[width=0.99\columnwidth]{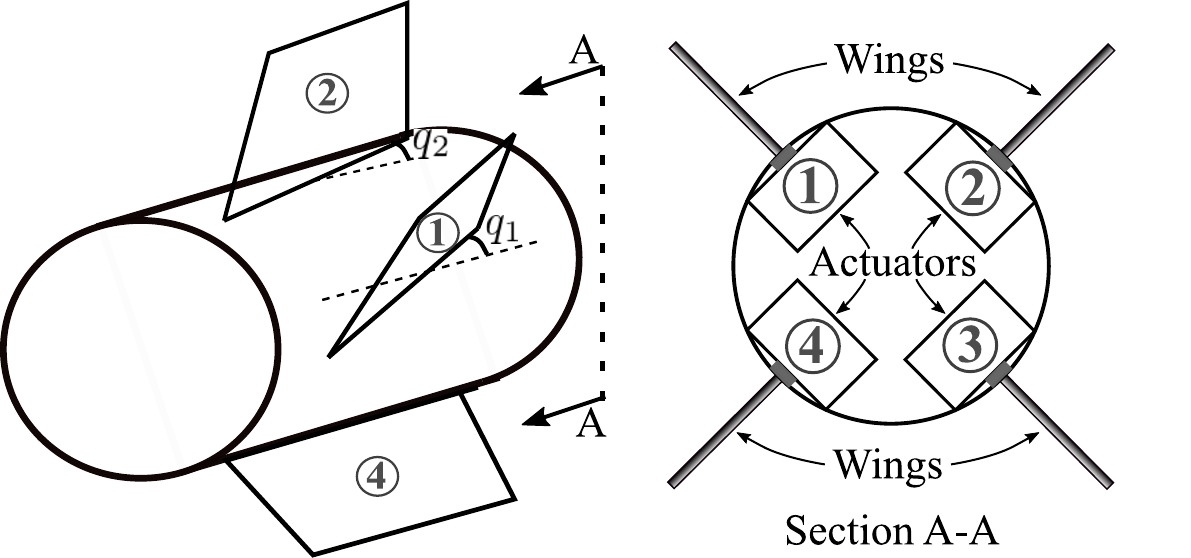}
		\caption{Missile fin actuation system and its section view.}
		\label{fig:linearoptimalsim}
	\end{figure}
	
	For each fin actuation channel, we consider the mechanical system used throughout Sec.~\ref{sec:frequencydomain} whose dynamics can be represented by the transfer function (\ref{eq:universalexamplemodel}). On the other hand, interaction of channels due to limits of the common power supply can be modeled by formulating the entire system as a multi DoF mechanical system by choosing $q~= ~[{q}_1 \quad {q}_2 \quad {q}_3 \quad {q}_4]^T$ as generalized coordinates, $x~=~[{q}^T \quad \dot{{q}}^T]^T$ as the system state, and $\upsilon = [u_1, u_2, u_3, u_4]^T$ as a torque input vector, which yields a linear system (\ref{eq:optimalcontrol_statespace}) in the state-space form with
	\begin{equation*}
	\begin{aligned}
	\!\!F_c\!=\!\!\!\left[\!\!\!\!
	\begin{array}{cc}
	0_{4 \times 4} & I_{4 \times 4} \\
	0_{4 \times 4} & (-d/m)\cdot I_{4 \times 4} 
	\end{array}\!\!\!\!
	\right],
	H_c\!=\!\!\!\left[
	\begin{array}{c}
	0_{4 \times 4}\\
	(1/m) \cdot I_{4 \times 4}
	\end{array}\!\!\!\!
	\right]\!, \text{ and } g_c\!=\!0.
	\end{aligned}
	\end{equation*}
	Three optimal control strategies are considered for this multi-DoF system. Controllers adopt the same  model for constraints 1 and 2, i.e., an inequality $A_{\rm neq} \Upsilon_N \leq b_{\rm neq}$, whereas they use different models for the constraint 3 (i.e., power supply limit). In particular, the following controllers are considered :	\begin{table}[!t]
		\centering
		\caption{Numerical values of actuator and controller parameters.}
		\label{tab:CASexample}
		  \setlength\extrarowheight{-3pt}
		
		\begin{tabular}{c||l}
			{\bf Parameter} & {\bf Value} \\ \hline
			$m$ & 1 $\text{kg} \,\text{m}^2$ \\ 
			$d$ & 0.05 $\text{N}\text{m}\,\text{s}/\text{rad}$ \\
			$\bar{u}$ & 180 $\text{N}\text{m}$ \\
			$u_{\rm stall}$ & 500 $\text{N}\text{m}$ \\
			$\dot{q}_{\max}$ & 4 $\text{rad}/\text{s}$ \\
			$\bar{R}$ & $(0.0056)\cdot\left[1, 1, 1, 1\right]^T$ $\text{Ohm}\,\text{A}^2 /(\text{N\text{m}})^2$ \\
			$P_{\max}$ & 750 W \\
			$\Delta t$ & 0.001 s \\
			$N$ & 300 \\
			$\Lambda$ & ${\rm diag}(2,2,2,2,\Delta t, \Delta t, \Delta t, \Delta t)\cdot(0.5/(\Delta t)^2)$ \\
			$\Phi$ & ${\rm diag}(1,1,1,1)$	\\
			$\Lambda_f$ & ${\rm diag}(0.1,0.1,0.1,0.1,\Delta t, \Delta t, \Delta t, \Delta t)\cdot(10/(\Delta t)^2)$	
		\end{tabular}
	\end{table}	
	\begin{enumerate}[C1:]
		\item The architecture (\ref{eq:mpc_optimproblem}) unifying a controller and a dynamic allocation strategy is evaluated first. To that end, the limit on aggregate power consumption (i.e., Constraint 2) is modeled by a quadratic constraint given in (\ref{eq:mpc_powerlimineq}), hence providing a solution to both control and dynamic power allocation problems concurrently.
		\item Modifying the first controller by reducing the aggregate power limit to power limit of individual actuators, a new controller can be obtained. This controller can be formulated as
		\begin{equation}
		\begin{aligned}
		& \underset{\Upsilon_N}{\text{min}}
		& & J(\Upsilon) = c_z + z \Upsilon_N + \Upsilon_N^T \,Z \,\Upsilon_N\\
		& \text{s.t.}
		& & A_{\rm neq} \Upsilon_N \leq B_{\rm neq} \\
		& \quad
		& & [\dot{q}_n]_i [\upsilon_n]_i + [\bar{R}]_i\,[\upsilon_n]_i^2 \leq \bar{P}_{i}, \; \forall n \wedge \forall i \in \mathcal{N}_a
		\end{aligned}
		\label{eq:mpc_optimproblem2}
		\end{equation}			 
		where $\mathcal{N}_a$ denotes the set of actuated coordinates, and $[\quad]_i$ operator selects the $i^{th}$ actuator from this set. This controller requires a predetermined static power allocation solution to the equation $\sum_{i=1}^{4}\bar{P}_i~=~P_{\max}$. For simulations, since actuators are identical, the trivial solution $\bar{P}_i=P_{\max}/{\rm dim}(\mathcal{N}_a) \,\forall i \in \mathcal{N}_a$ is used.
		\item The last controller can be derived from the second controller by  modeling an actuator's power limit with the conservative approximation (\ref{eq:psat_approx}). This enables formulation of the power limit in terms of a linear constraint, which can be captured by the optimal controller 
		\begin{equation}
		\begin{aligned}
		& \underset{\Upsilon_N}{\text{min}}
		& & J(\Upsilon) = c_z + z \Upsilon_N + \Upsilon_N^T \,Z \,\Upsilon_N\\
		& \text{s.t.}
		& & A_{\rm neq} \Upsilon_N \leq B_{\rm neq} \\
		& \quad
		& & [\upsilon_n]_i \leq \bar{P}_{i} / \dot{q}_{{\rm max},i}, \; \forall n \wedge \forall i \in \mathcal{N}_a.
		\end{aligned}
		\label{eq:mpc_optimproblem3}
		\end{equation}							
	\end{enumerate}	
	\begin{figure}[!t]
		\centering
		\includegraphics[width=0.8\columnwidth]{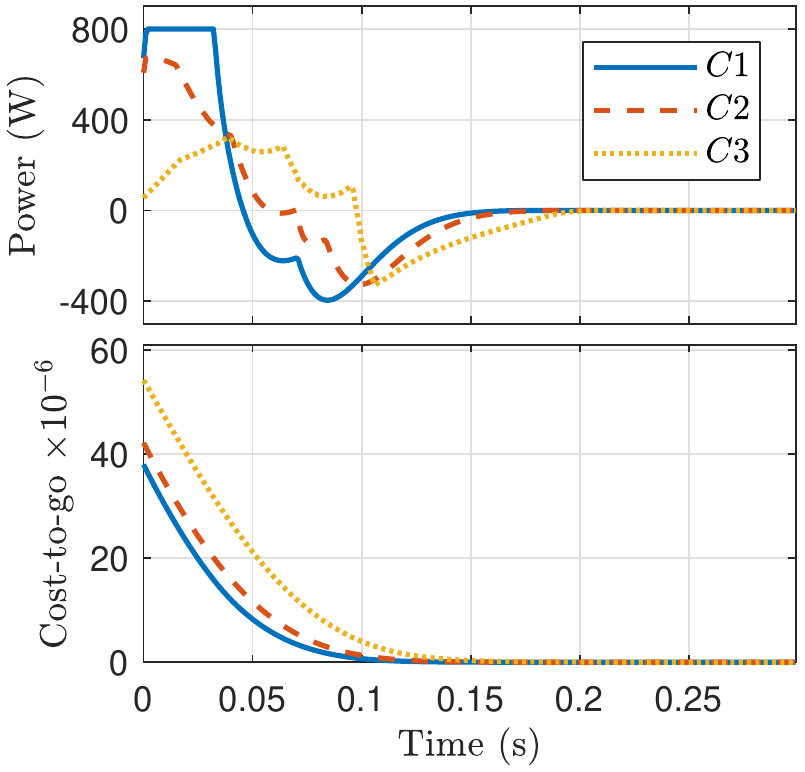}
		\caption{Cost-to-go function $J$ (bottom) scaled by a factor of $10^{-6}$ and total power consumption $\sum_{i=1}^{4}P_i$ of actuators (top), both as a function of time, for optimal controllers $C1$ (solid-blue), $C2$ (dashed-red), and $C3$ (dotted-yellow).}
		\label{fig:cas_example_PowerCost2go}
	\end{figure} Using numerical values given in Table~\ref{tab:CASexample} for the physical model and controllers, simulations are performed to evaluate controllers for an example task of driving the system state $x$ from the initial condition $x_0~=~\left[0.5, -0.16, 0.08, 0.28, 0, 0, 0, 0\right]^T$ to the target state $x^\star = 0$. Simulation results illustrated in Figures \ref{fig:cas_example_PowerCost2go} and \ref{fig:cas_example_PosTorque} indicate followings : 
	\begin{figure*}[!t]
		\includegraphics[width=\textwidth]{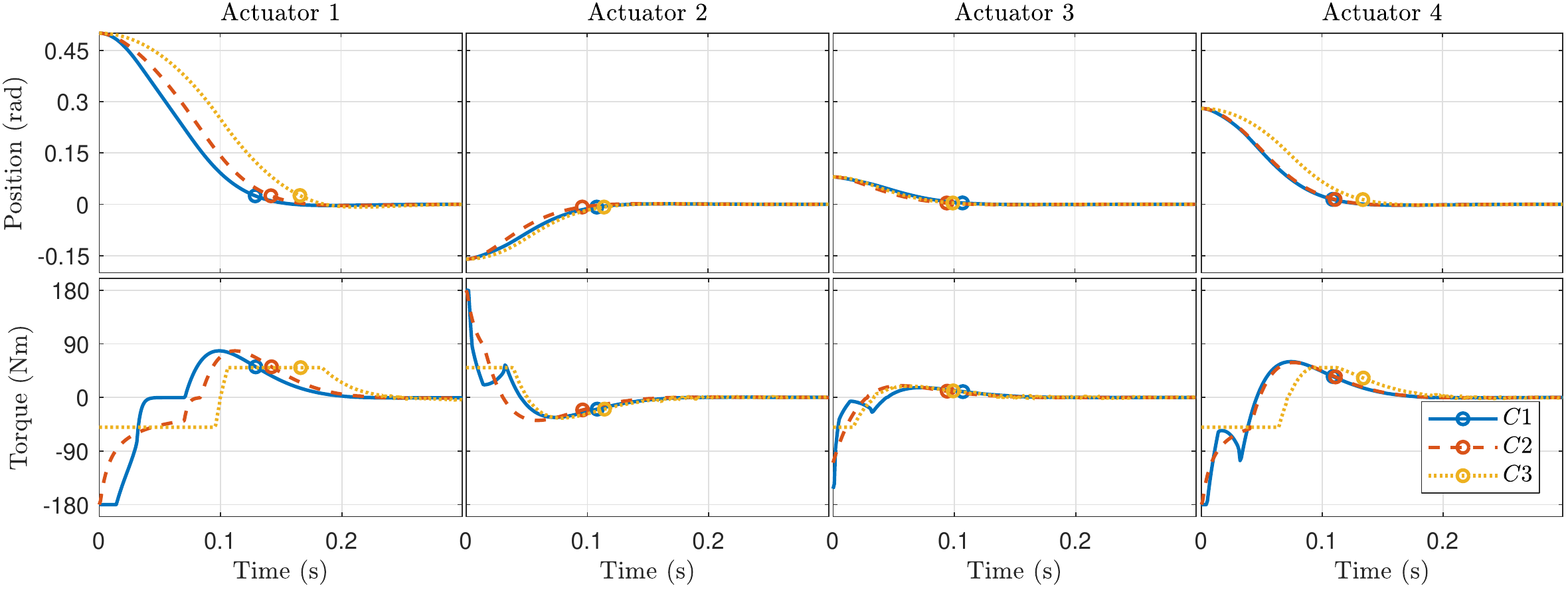}
		\caption{Position (top) and torque (bottom) trajectories of actuators whose settling times are marked with circles for optimal controllers $C1$ (solid-blue), $C2$ (dashed-red), and $C3$ (dotted-yellow). Settling time is defined as the time required for the response curve to reach and stay within 5\% of its initial value.}
		\label{fig:cas_example_PosTorque}
	\end{figure*}
	\begin{itemize}
		\item Controller $C3$, which adopts neither a dynamic allocation strategy nor an exact power model, performs the worst, manifesting itself as the slowest settling time to target location and as the highest cost function throughout all this time since it cannot deliver more than half of the available power as shown in Fig.~\ref{fig:cas_example_PowerCost2go}. This result is in favour of the main idea proposed in this paper, which is to exactly model power supply limit in controllers.
		\item Controller $C1$, i.e., the unified strategy proposed in this section, outperforms other strategies by providing the minimum cost-to-go function, as depicted in Fig.~\ref{fig:cas_example_PowerCost2go}. This is achieved by  allocating larger portion of the available power to particular joints which contribute to overall tracking performance more than other joints. This is in contrast with the static allocation which reserves equal amount of power for each joint. In particular, there are two mechanisms that can act as a medium for that strategy: 1) Residual power of joints with lower priority are channelled to joints with higher priority without violating the aggregate power limit so as to benefit from the power supply as much as possible. This is illustrated in Fig.~\ref{fig:cas_example_PowerCost2go} where controller $C1$ uses 100\% of the permissible power source in the first 50 ms of the simulation. 2) Local performance of joints with lower priority is sacrificed to improve overall performance, measured by the cost-to-go function, as illustrated in Fig.~\ref{fig:cas_example_PosTorque} where controller $C1$ tolerates a slight delay in the response of actuator 3 to improve response of actuator 1. 
	\end{itemize}
\end{exmp}

\subsection{Nonlinear Systems}
\label{sec:nonlinearmpc}
Even though dynamic allocation increases control efficiency for linear systems as shown in the previous section, it requires solving a non-convex QCQP optimization problem  which is not computationally efficient. Furthermore, if this finite-horizon optimization problem is generalized to nonlinear systems, this situation gets even worse with non-quadratic and nonlinear objective and constraint functions. In this subsection, an alternative approach based on a control lyapunov function (CLF) in conjunction with a pointwise-minimum-norm control \cite{freeman_kokotovic.ACC1995,freeman_kokotovic.SIAM1996,ames_grizzle_sreenath.TAC2014} is proposed to obtain a more practical unified algorithm.

Assuming that power limit nonlinearity is satisfied, dynamics in (\ref{eq:eom_nonlinear}) can be cast into state-space form
\begin{equation}
\dot{x} = f(x) + g(x) u 
\label{eq:jointstatespace}
\end{equation}
\begin{flalign*}
\begin{aligned}
\text{with } x &= \begin{bmatrix}
q \\
\dot{q}
\end{bmatrix}, \quad g(x)=\begin{bmatrix}
0 \\
M^{-1}(q)
\end{bmatrix}, \text{ and} \\
f(x) &= \begin{bmatrix}
\dot{q} \\
- M^{-1}(q) \left(C(q,\dot{q}) \dot{q} + D \dot{q} + G(q)\right)
\end{bmatrix}.
\end{aligned}&&&
\end{flalign*} 
Now, consider a task $y(x)$ of joint positions with the goal to construct a feedback control law embedding a dynamic power allocation strategy and being compatible with constraints (e.g., torque limits and a power supply limit) such that dynamics of positional error $\tilde{y} := y(x) - y^\star(t)$ between the output function $y(x)$ and a desired trajectory $y^\star(t)$ is rendered GAS with an equilibrium point $e:=\begin{bmatrix}\tilde{y}, & \dot{\tilde{y}}\end{bmatrix} = 0$. The proposed control design is based on prior knowledge of a CLF which is defined as a continuously-differentiable, proper (i.e., radially-unbounded), and positive-definite function $V : \mathbb{R}^{n}\to \mathbb{R}_+$ with $n={\rm dim}(e)$ such that 
\begin{equation}
\inf_{u}\left[\dot{V}(e) \right] < 0.
\label{eq:clfratedefn}
\end{equation}
In other words, the existence of a CLF for tracking $y^\star(t)$ is equivalent to the existence of a stabilizing controller which satisfies the inequality (\ref{eq:clfratedefn}). Such a controller and a CLF pair can be found by feedback linearization \cite{ames.ACC2016} systematically as follows : Considering $\tilde{y}(x,t)$ as the output function, input-output dynamics have a relative degree of two, hence taking the form
\begin{equation}
\ddot{\tilde{y}} = L_f^2 \tilde{y} + \left(L_g L_f \tilde{y}\right) u - \ddot{y}^\star
\label{eq:inputoutputdynamics}
\end{equation}
where $L_{\mathcal{X}} \mathcal{T}$ denotes the Lie derivative of a tensor field $\mathcal{T}$, which can be a scalar function, a vector field, or a one-form, along the flow of a vector field $\mathcal{X}$ \cite{lee.BOOK2012}. Observe that error dynamics (\ref{eq:inputoutputdynamics}) is a time-varying nonlinear system which can be represented in state-space form
\begin{equation}
\dfrac{d}{d\,t}\underbrace{\begin{bmatrix}
	\tilde{y} \\
	\dot{\tilde{y}}
	\end{bmatrix}}_{e} = 
\underbrace{\begin{bmatrix}
	\dot{\tilde{y}} \\
	L_f^2 \tilde{y} - \ddot{y}^\star(t)
	\end{bmatrix}}_{\tilde{f}(t, e)}
+
\underbrace{\begin{bmatrix}
	0_{n \times 1} \\
	L_g L_f \tilde{y}
	\end{bmatrix}}_{\tilde{g}(t,e)} u.
\label{eq:errorstatespace}
\end{equation}
As an illustrative example, consider the full joint-space control i.e., $y(x) = q$. 
In this case, input-output system (\ref{eq:errorstatespace}) becomes a trivial extension of the original state-space system (\ref{eq:jointstatespace}) with a time-dependent term as
\begin{equation*}
\begin{aligned}
\tilde{f}(t,e) &= f(x) - \begin{bmatrix}(\dot{q}^\star)^T & (\ddot{q}^\star)^T\end{bmatrix}^T \\
\tilde{g}(t,e) &= g(x),
\end{aligned}
\end{equation*}which shows that (\ref{eq:errorstatespace}) is sufficiently general to express a variety of tasks. Therefore, a feedback-linearizing control law
\begin{equation}
u = \left(L_g L_f y(x)\right)^{-1} \left(\tilde{u} + \ddot{y}^\star - L_f^2 y (x)\right),
\label{eq:uinputoutput}
\end{equation}
with an auxiliary controller $\tilde{u}$ transforms nonlinear dynamics of the system (\ref{eq:errorstatespace}) to those of a linear system
\begin{equation*}
\dfrac{d\,e}{d\,t} = \begin{bmatrix}
0_{n \times n} & I_{n \times n} \\
0_{n \times n} & 0_{n \times n}
\end{bmatrix} e
+ 
\begin{bmatrix}
0_{n \times n} \\
I_{n \times n}
\end{bmatrix} \tilde{u}.
\end{equation*}
which can be globally asymptotically stabilized by choosing auxiliary part of the controller as
\begin{equation}
\tilde{u} = - K_p\,\tilde{y} - K_d\,\dot{\tilde{y}}.
\label{eq:pdinputoutput}
\end{equation}
Resulting closed-loop error dynamics can be expressed in state-space form as
\begin{equation}
\dfrac{d\, e}{d\,t} = 
\underbrace{\begin{bmatrix}
	0_{n \times n} & I_{n \times n} \\
	-K_p & -K_d
	\end{bmatrix}}_{A_{cl}} \, e
\label{eq:errorcl}
\end{equation}
A candidate CLF for the system in (\ref{eq:errorstatespace}) is $V(e) = e^T \mathcal{P} e$ where $\mathcal{P}$ is a symmetric positive-definite matrix solution to the Lyapunov equation
\begin{equation}
A_{cl}^T \mathcal{P} + \mathcal{P} A_{cl} + \mathcal{W} = 0
\label{eq:riccatti}
\end{equation}
for some symmetric positive-definite matrix $\mathcal{W}$. Then, derivative of $V$ along trajectories of the closed-loop system (\ref{eq:errorcl}) satisfies
\begin{equation}
\dot{V} = -e^T \mathcal{W} e \leq -\varepsilon V  \leq 0,
\label{eq:clfrate}
\end{equation}
verifying the existence of a exponentially stabilizing CLF for the system (\ref{eq:errorstatespace}) 
such that
\begin{equation*}
V(e(t)) \leq V(e(0))\,e^{{\displaystyle-\varepsilon t}}
\end{equation*}
with $\varepsilon := \lambda_{\min}(\mathcal{W}) / \lambda_{\max}(\mathcal{P})$. 

The above procedure provides a development framework to build controllers for nonlinear mechanical systems despite its well-known drawbacks \cite{freeman_kokotovic.ACC1995,ames_grizzle_sreenath.TAC2014}. In particular, a feedback linearization controller may apply unnecessarily large control signals by blindly cancelling nonlinearities that might be beneficial to convergence of the system to the desired trajectory. Moreover, that increase in control effort may conflict with physical constraints such as torque limits. Motivated from these disadvantages, \cite{freeman_kokotovic.ACC1995,freeman_kokotovic.SIAM1996}, and \cite{primbs_doyle.AJC1999} proposed a new family of controllers based on pointwise solution to the QP problem
\begin{equation}
\begin{aligned}
& \underset{u}{\text{min}}
& & (u-u_0)^T \Phi (u-u_0)\\
& \text{s.t.}
& & L_{\tilde{f}} (e^T \mathcal{P} e) + L_{\tilde{g}} (e^T \mathcal{P} e) \, u \leq - e^T \mathcal{W} e
\end{aligned}
\label{eq:qpclf_basic}	
\end{equation}
which seeks a minimum norm solution to control signals around a baseline value $u_0$ weighted by a matrix $\Phi=\Phi^T > 0$ under the stabilization constraint expressed as an upper bound on time-derivative of the CLF (\ref{eq:clfrate}). The baseline control signal can be chosen in a variety of ways. For instance, $u_0 = 0$ can be used to minimize the instantaneous control actions or $u_0 = u^\star$ to enforce a particular control policy where $u^\star$ is computed by adding output of a baseline feedback law and feedforward torques for tracking the desired trajectory or even $u_0 = u_{k-1}$ to minimize the change in control actions. Note that the QP problem is solved at each sampling time, and controller output is updated accordingly. This strategy (hereinafter referred to as CLF-QP) has been successfully implemented for locomotion control of bipedal robots \cite{sreenath_ames_grizzle.IEEE2015} and for adaptive cruise control \cite{tabuada_grizzle_ames.ACC2015}. It can also serve as a framework to formulate a unified algorithm combining controls and dynamic power allocation by incorporating physical constraints into the optimization framework (\ref{eq:qpclf_basic}). As explained in Subsection \ref{sec:linearmpc}, there are two types of constraints : 1)~Affine constraints such as torque limits and speed/torque curve of each joint. These constraints can be enforced by a linear inequality $A_{\rm neq}(q,\dot{q})\,u \leq b_{\rm neq}(q,\dot{q})$ that needs to be satisfied at all times. 2) Instantaneous power supply limit which can be modeled as
\begin{equation}
u^T \Omega\,u + \dot{q}^T u \leq P_{\max}.
\label{eq:clfpowerlim}
\end{equation}
With these constraints, the QP formulation of a unified algorithm takes the form
\begin{equation}
\begin{aligned}
& \underset{u}{\text{min}}
& & (u-u_0)^T \Phi (u-u_0)\\
& \text{s.t.}
& & L_{\tilde{f}} (e^T \mathcal{P} e) + L_{\tilde{g}} (e^T \mathcal{P} e) \, u \leq - e^T \mathcal{W} e \\
& 
& & A_{\rm neq}(q,\dot{q})\,u \leq b_{\rm neq}(q,\dot{q}) \\
& 
& & u^T \Omega\,u + \dot{q}^T u \leq P_{\max}.
\end{aligned}
\label{eq:qpclf_final}	
\end{equation}
We now state the main result of this subsection.
\begin{thm}
	The unified algorithm based on CLF-QP formulation (\ref{eq:qpclf_final}) is a convex optimization problem.
\end{thm}
\begin{proof}
	Assume that position and velocity measurements are available. The proof is done separately for the following two cases :  i) $\bar{R} \neq 0$ ii) $\bar{R} = 0$, corresponding to full model of power limit including electrical losses and reduced model without electrical losses, respectively.
	\begin{enumerate}[i)]
		\item In this case, (\ref{eq:qpclf_final}) is a QCQP problem since the power limit constraint is quadratic due to nonzero $\Omega$. The convexity follows from the fact that $\Phi$ is a symmetric positive-definite matrix as defined in (\ref{eq:qpclf_basic}) and that $\Omega$ is a diagonal matrix with positive entries qualifying it as positive-definite symmetric matrix.
		\item In this case, (\ref{eq:qpclf_final}) reduces to a standard QP without a quadratic constraint since the power supply limit becomes an affine constraint with $\Omega=0$ due to $\bar{R}=0$. Convexity is automatically obtained as a result of this situation.
	\end{enumerate}
\end{proof}
Finally, observe that (\ref{eq:qpclf_final}) might be an infeasible QP problem since stabilization and physical constraints may conflict with each other. To ensure feasibility, we relax the QP problem by introducing a slack variable into CLF rate inequality such that intersection of feasible regions corresponding to constraints is never empty. The relaxed version of the unified algorithm can be formulated as 
\begin{equation}
\begin{aligned}
& \underset{u, p_s}{\text{min}}
& & (u-u_0)^T \Phi (u-u_0) + c_s p_s^2\\
& \text{s.t.}
& & L_{\tilde{f}} (e^T \mathcal{P} e) + L_{\tilde{g}} (e^T \mathcal{P} e) \, u \leq - e^T \mathcal{W} e + p_s\\
& 
& & A_{\rm neq}(q,\dot{q})\,u \leq b_{\rm neq}(q,\dot{q}) \\
& 
& & u^T \Omega\,u + \dot{q}^T u \leq P_{\max}.
\end{aligned}
\label{eq:qpclf_relaxed}	
\end{equation}
with a slack variable $p_s$ and a penalty coefficient $c_s$ associated with $p_s$. 
\begin{exmp}
	Consider the example in Sec.~\ref{sec:pbc} where a two-link manipulator illustrated in Fig.~\ref{fig:twolink} is used to evaluate a passivity-based controller. This time, the robot example is used to evaluate the proposed CLF-QP-based approach that combines the dynamic allocation and a task controller. To that end, we assume that the robot's physical parameters given in Table~\ref{tab:pbcsim} are the same. As for the controller, in contrary to that example, PD gains $K_p$ and $K_d$ of the feedback-linearizing controller (\ref{eq:uinputoutput}) are chosen as
	\begin{equation*}
	\begin{aligned}
	K_p &= \omega_n^2 \\
	K_d &= 2 \zeta \omega_n,
	\end{aligned}
	\end{equation*}
	with $\omega_n = 2 \pi (2.2)\cdot I_{2\times2}$ and $\zeta = (\sqrt{3}/2)\cdot I_{2\times2}$. This defines the closed-loop system (\ref{eq:errorcl}) whose 
	CLF $V=e^T \mathcal{P} e$ can be parametrized as
	\begin{equation*}
	\mathcal{P} = \begin{bmatrix}
	2\zeta \omega_n^2 & 2 \omega_n \sqrt{1-\zeta^2} \\
	2 \omega_n \sqrt{1-\zeta^2} & 2\zeta
	\end{bmatrix}.
	\end{equation*}
	Furthermore, electrical losses of the motor and two physical constraints are included into the system model as follows :
	\begin{enumerate}
		\item The motor resistance normalized by torque constant is chosen as $\bar{R} = \begin{bmatrix}0.0833,&0.222\end{bmatrix} \rm{mOhm}\,\rm{A}^2 /(\rm{N}\rm{m})^2$.
		\item Joint torque limits $-\bar{u}_i \leq u_i \leq \bar{u}_i$ are enforced as 
		\begin{equation*}
		\underbrace{\begin{bmatrix}
			1 & 0 \\
			-1 & 0 \\
			0 & 1 \\
			0 & -1
			\end{bmatrix}}_{\displaystyle A_{\rm neq}}
		\underbrace{\begin{bmatrix}
			u_1 \\
			u_2
			\end{bmatrix}}_{\displaystyle u}
		\leq
		\underbrace{\begin{bmatrix}
			\bar{u}_1 \\
			\bar{u}_1 \\
			\bar{u}_2 \\
			\bar{u}_2
			\end{bmatrix}}_{\displaystyle b_{\rm neq}}
		\end{equation*}
		with $\bar{u}=\begin{bmatrix}2000, & 1000\end{bmatrix} \text{Nm}$.
		\item Power supply limit is enforced as (\ref{eq:clfpowerlim}) with $P_{\max} = 1 \text{ kW}$.
	\end{enumerate}
	
	Based on this system model, three controllers are implemented in a simulation environment for validation and performance comparison.
	\begin{enumerate}[C1:]
		\item The first controller is CLF-QP-based unified algorithm  (\ref{eq:qpclf_relaxed}) proposed in this subsection with a zero baseline control signal 
		$u_0 = \begin{bmatrix}0, & 0\end{bmatrix}$, identity weighting matrix $\Phi~=~I_{2\times2}$ and the penalty coefficient $c_s~=~5\cdot10^4$, respectively.
		\item The second controller is again a CLF-QP-based unified algorithm but with static allocation as opposed to the controller $C1$ employing a dynamic allocation. This difference can be captured by modifying QP problem (\ref{eq:qpclf_relaxed}) as
		\begin{equation*}
		\begin{aligned}
		& \underset{u, p_s}{\text{min}}
		& & (u-u_0)^T \Phi (u-u_0) + c_s p_s^2\\
		& \text{s.t.}
		& & L_{\tilde{f}} (e^T \mathcal{P} e) + L_{\tilde{g}} (e^T \mathcal{P} e) \, u \leq - e^T \mathcal{W} e + p_s\\
		& 
		& & A_{\rm neq}(q,\dot{q})\,u \leq b_{\rm neq}(q,\dot{q}) \\
		& 
		& & \bar{R}_1 u_1^2 + \dot{q}_1 u_1 \leq P_{\max}/2 \\
		& 
		& & \bar{R}_2 u_2^2 + \dot{q}_2 u_2 \leq P_{\max}/2.
		\end{aligned}
		\end{equation*}		
		The CLF and other parameters of this controller are the same as those used by controller $C1$.
		\item The third strategy is the standard feedback-linearization controller described in this section. 
	\end{enumerate}
	Using initial conditions $x(0) = \begin{bmatrix}-\pi/2&0&0&0\end{bmatrix}$, simulations are conducted to evaluate these three controllers in a joint-space regulation task toward a target position $q^{\star}~=~\begin{bmatrix}\pi/2&0\end{bmatrix}$. This corresponds to the output function $y(x)~=~q - q^{\star}$ which further defines the control procedure explained in this section. Simulation results illustrated in Figures \ref{fig:2link_example_PosTorque} and \ref{fig:2link_example_PowerCost2go} show the followings : Controller $C_1$ demonstrates the best performance. This can be seen from the position trajectories of joints as well as the CLF. In particular, joint 1 reaches the target angle in the shortest time while providing the minimum amount of deviation for joint 2 from its desired position. This is achieved by channelling greater amount of power to joint 1 than other controllers without exceeding torque limits. On the other hand, static allocation strategies are comparable to each other. This manifests itself as CLF trends with similar time constants. However, the controller $C2$ produces a less deviation in joint 2 whereas the controller $C3$ provides a shorter settling time for joint 1. While doing so, the controller $C2$ requires less control effort because of the fact that it provides the same negative CLF rate with the the minimum torque.
	
	\begin{figure}[!h]
		\centering
		\includegraphics[width=0.7\columnwidth]{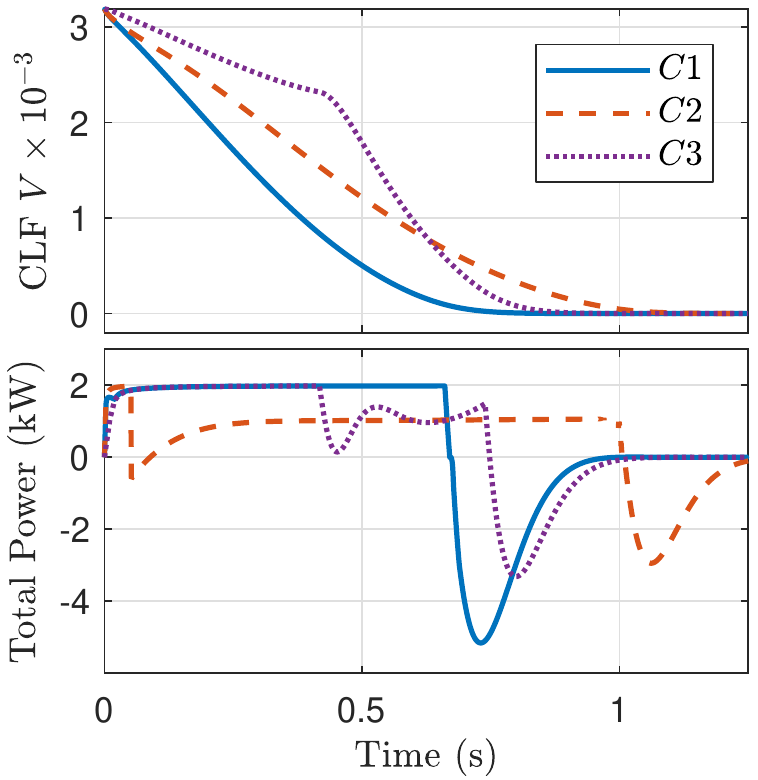}
		\caption{Value of CLF $V$ (top) scaled by a factor of $10^{-3}$ and total power consumption $\sum_{i=1}^{2}P_i$ of joints (bottom), both as a function of time, for controllers $C1$ (solid-blue), $C2$ (dashed-red), and $C3$ (dotted-purple).}
		\label{fig:2link_example_PowerCost2go}
	\end{figure}
	
	\begin{figure}[!h]
		\centering
		\includegraphics[width=0.8\columnwidth]{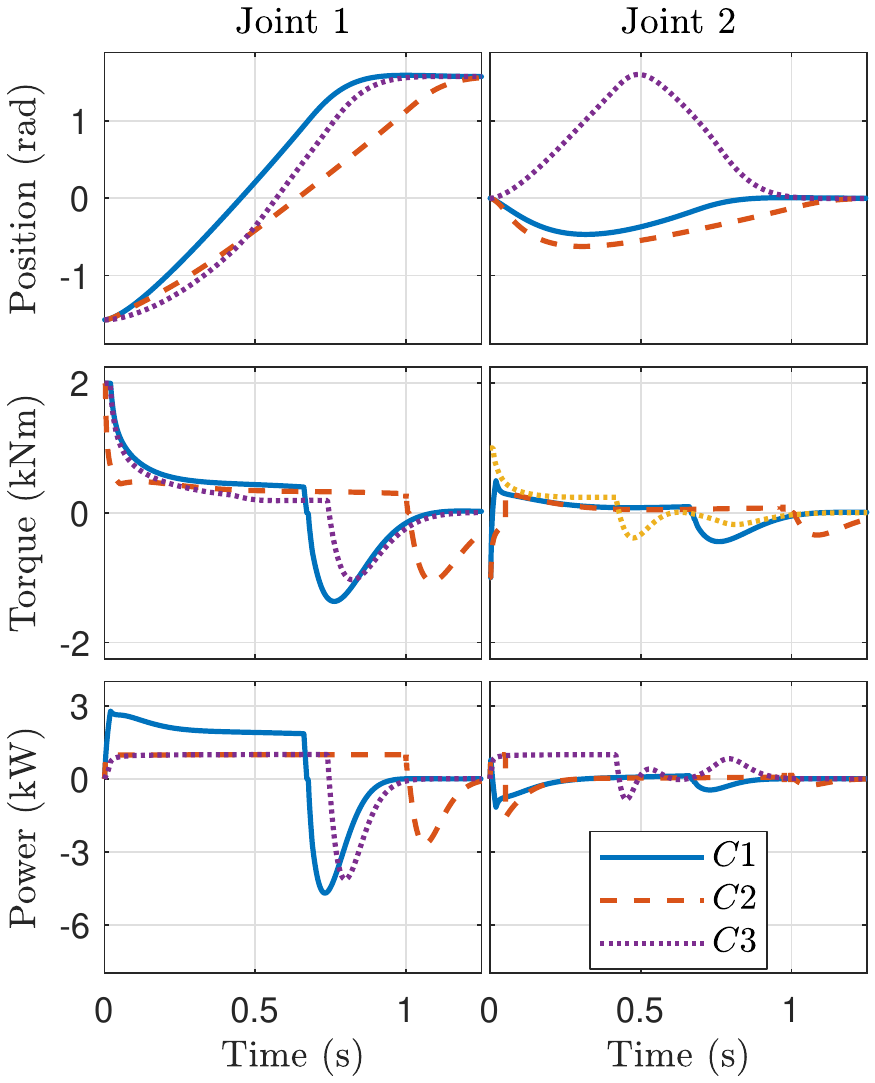}
		\caption{Position (top), torque (middle), and power (bottom) trajectories of joint 1(left) and joint 2 (right) for controllers $C1$ (solid-blue), $C2$ (dashed-red), and $C3$ (dotted-purple).}
		\label{fig:2link_example_PosTorque}
	\end{figure}
\end{exmp}

\section{Preliminary Experimental Results}
\label{sec:experiments}
In this section, we present results of preliminary experiments conducted on a single DoF electromechanical actuator that we built to comparatively evaluate two controllers with exact and approximate models of power supply limit, respectively, as detailed in Sec.~\ref{sec:dynamiccontrol}. The experimental setup illustrated in Fig.~\ref{fig:experimentalsetup} includes the actuator, a power-supply, and a test rig, which is a self-contained system capable of applying controlled external loads to and measuring output position and power consumption of the actuator. The actuator consists of a drivetrain, a PMSM with an encoder, and control/driver electronics based on an FPGA with a generic six bridge power amplifier. The electronics run a field oriented controller for torque control and a PID-type trajectory tracking controller with CI that is designed according to the procedure in Sec.~\ref{sec:dynamiccontrol}. Parameters of the actuator and other details of the experimental setup are given in Table~\ref{tab:experimentalsetup}.

\begin{table}[!b]
	\centering
	\caption{Parameters of the actuator, controllers, and the power supply in the experimental setup.}
	\label{tab:experimentalsetup}
	\setlength\extrarowheight{-5pt}
	
	\begin{tabular}{r|l}
		{\bf Parameter} & {\bf Value} \\ \hline
		Motor & Maxon EC-4pole 30 \\ 
		Phase-to-Phase Resistance & 0.1 $\Omega$ \\
		Motor Efficiency & 95 \% \\ 
		Voltage & 24 V \\
		Inertia & 1 $\text{kg}\,\text{m}^2$ \\ 
		Viscous damping & 0.05 $\text{N}\,\text{m}/(\text{rad/s})$ \\
		Torque control sampling rate & 20 kHz \\
		Torque control bandwidth & 500 Hz \\
		Maximum speed $\dot{q}_{\max}$ & 4 rad/s \\
		Peak driver current $I_{\max}$ & 32 A \\
		Torque constant $k_t$ & 6 $\text{N}\,\text{m}/\text{A}$ \\
		Power supply limit $P_{\max}$ & 400 W \\
		Position control sampling rate & 2 kHz \\
	\end{tabular}
\end{table}	

Two controllers are implemented on the actuator for experimental evaluation. While both controllers are of the PID-type, they use different models of the power supply limit nonlinearity, as mentioned above. The first algorithm, which is considered as a baseline controller and abbreviated as $C1$ from now on, employs the exact model (\ref{eq:psat_simple}), which ignores the electrical losses of the motor since its efficiency is close to unity, in combination with the torque saturation nonlinearity corresponding to the peak current of the driver to find a time-varying instantaneous torque limit as
\begin{equation*}
u_{\max}(t) = \min(I_{max} k_t, P_{\max}/\dot{q}(t)).
\end{equation*}
On the other hand, the second controller which will be referred as $C2$ from now on, uses the approximate model (\ref{eq:psat_approx}), which defines a standard torque saturation ${\rm sat}(u, u_{\max})$ with
\begin{equation*}
u_{\max} = \min(I_{max} k_t, P_{\max}/\dot{q}_{\max}) = P_{\max}/\dot{q}_{\max} = 100 \,\text{N}\,\text{m.}
\end{equation*}
{
	\vspace{-1cm}
\begin{figure}[!t]
	\centering
	\includegraphics[width=0.8\columnwidth]{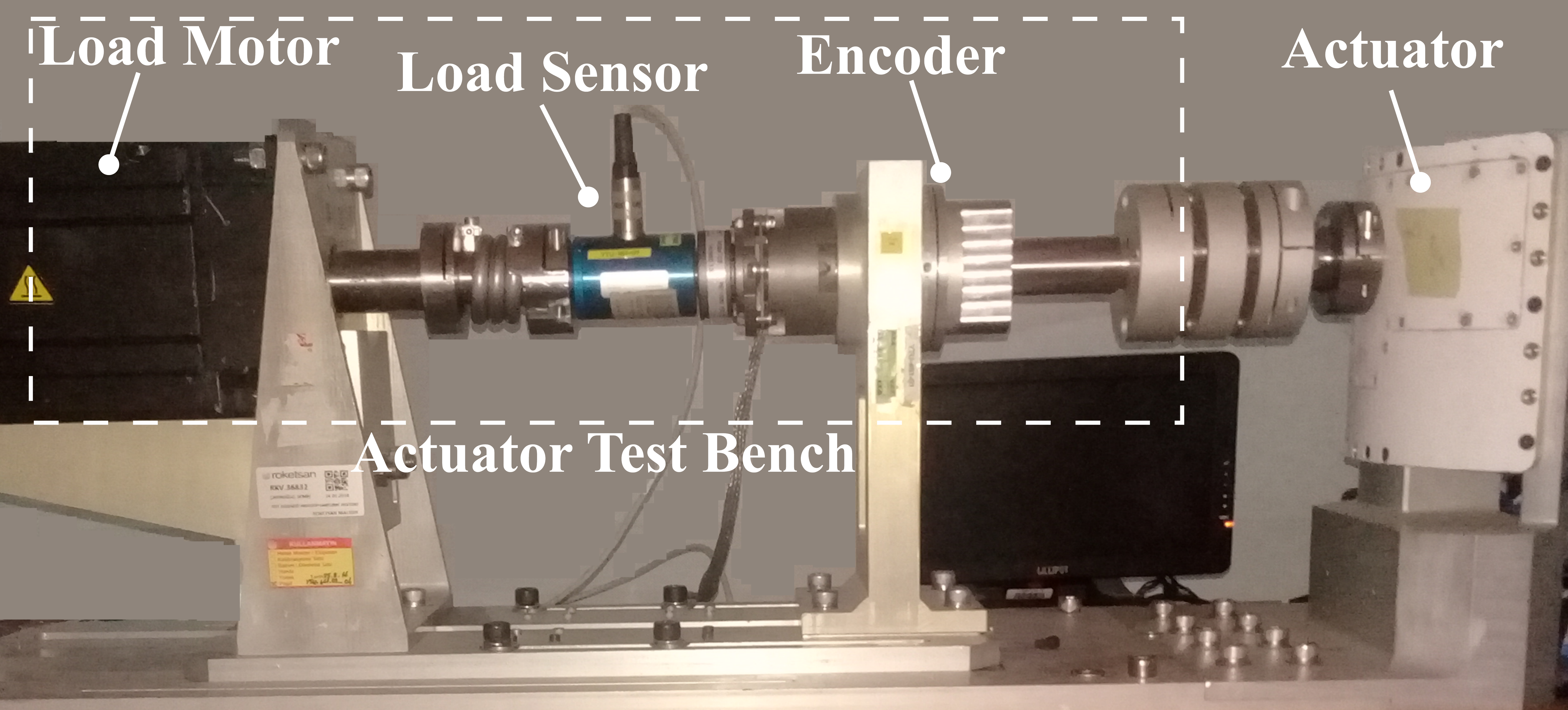}
	\caption{Experimental setup.}
	\label{fig:experimentalsetup}
\end{figure}

Two types of experiments are conducted to validate controllers and to compare their resulting closed-loop control performance :

1) Time-domain experiments : Step-input in position-command with different amplitudes are applied to controllers whose time-domain responses are illustrated in Fig.~\ref{fig:timedomain_experiment} from which a comparison can be made in terms of their maximum overshoots and settling times. In particular, as is evident from the results reported in Table~\ref{tab:timedomain_experiment}, the baseline controller $C1$, which uses the exact model of power limitation, outperforms the controller with the approximate model clearly by providing nearly half settling-time for the command with the highest amplitude without any significant increase in the overshoot whereas responses of controllers become less different as amplitude decreases. This is due to the fact that larger torques are needed to achieve the desired response to step commands with greater amplitudes, which controller $C2$ cannot deliver because of the conservative limitations on admissible torques defined by the approximate power limit model.}
\begin{figure}[!h]
	\centering
	\includegraphics[width=0.8\columnwidth]{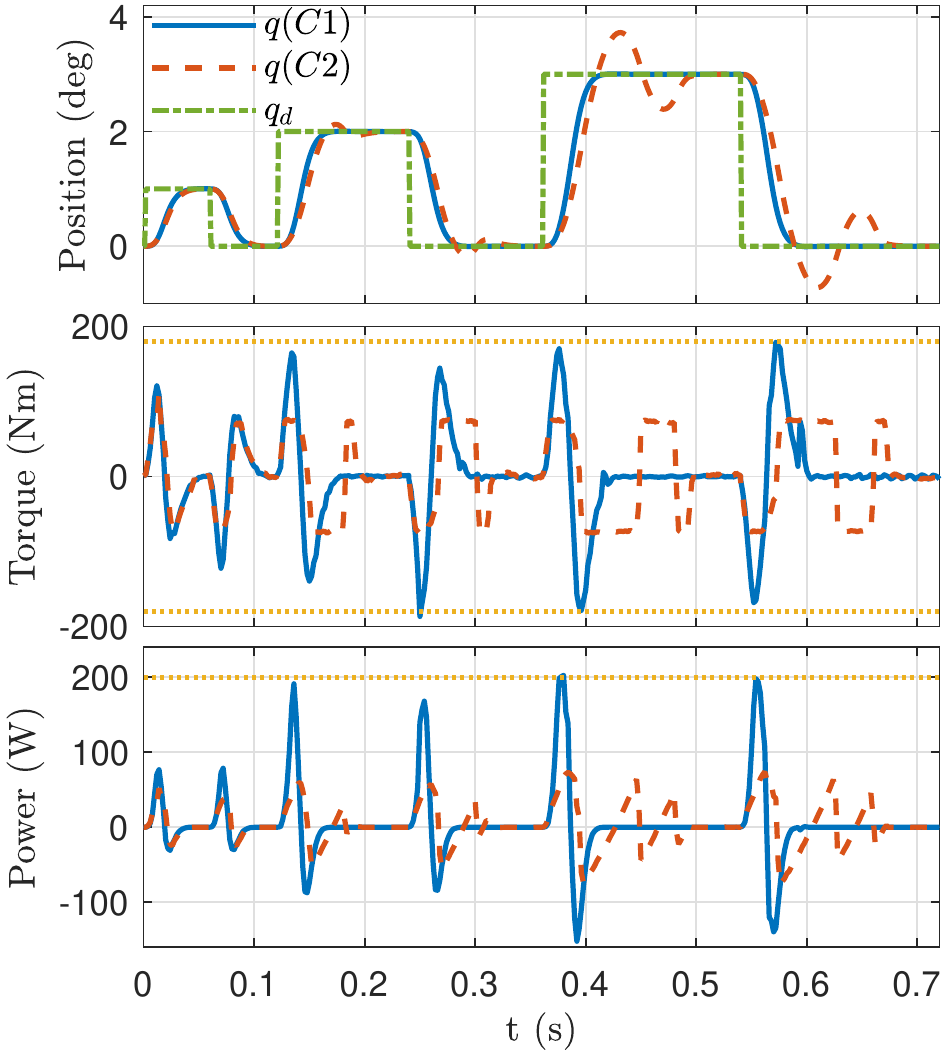}
	\caption{Power consumptions (bottom), motor torques (middle) and position responses (top) of controllers $C1$ (solid blue line) and $C2$ (dashed red line) to $1^\circ$,$2^\circ$, and $3^\circ$ step commands (dash dotted green line in the top figure). Torque limits due to the driver's peak current rating and power supply limit are marked with the dotted orange lines.}
	\label{fig:timedomain_experiment}
\end{figure}

\begin{table}[!h]
	\centering
	\caption{Settling time and Percent Overshoot (PO) values of controllers for different step commands.}
	\label{tab:timedomain_experiment}
	\setlength\extrarowheight{-5pt}
	\begin{tabular}{r|c|c|c|}
		{\bf Controller} & {\bf Amplitude} & {\bf Settling Time} & {\bf PO} \\ \hline\hline
		\multirow{ 3}{*}{$C1$} & $1^\circ$ & 0.036 & 0 \\
		& $2^\circ$ & 0.039 & 0 \\
		& $3^\circ$ & 0.043 & 0 \\ \hline
		\multirow{ 3}{*}{$C2$} & $1^\circ$ & 0.038 & 0 \\
		& $2^\circ$ & 0.051 & 6.0 \\
		& $3^\circ$ & 0.071 & 24.30 \\ 	\hline	
	\end{tabular}
\end{table}
	
2) Frequency-domain experiments:  We compare the closed-loop frequency response of controllers by stimulating the system with a chirp command which has a unit amplitude and frequency content in the range between $1 Hz$ and $24 Hz$ and by measuring the resulting frequency response functions, which are illustrated in Figure~\ref{fig:frequencydomain_experiment_bode}. The Bode plots show that controller $C1$ yields a better closed-loop tracking performance by providing higher magnitude and less phase delay in the high frequency region (especially above 10 Hz). Furthermore, the power consumption of both controllers given in Figure~\ref{fig:frequencydomain_experiment_power} shows that the controller $C1$ channels as much power as permissible to the motor without exceeding the power supply limit by reacting quickly to the changes in the measurements. This result validates not only the effectiveness of the proposed control approach but also its safety.
	\vspace{-0.5cm}
\begin{figure}[!h]
	\centering
	\includegraphics[width=0.77\columnwidth]{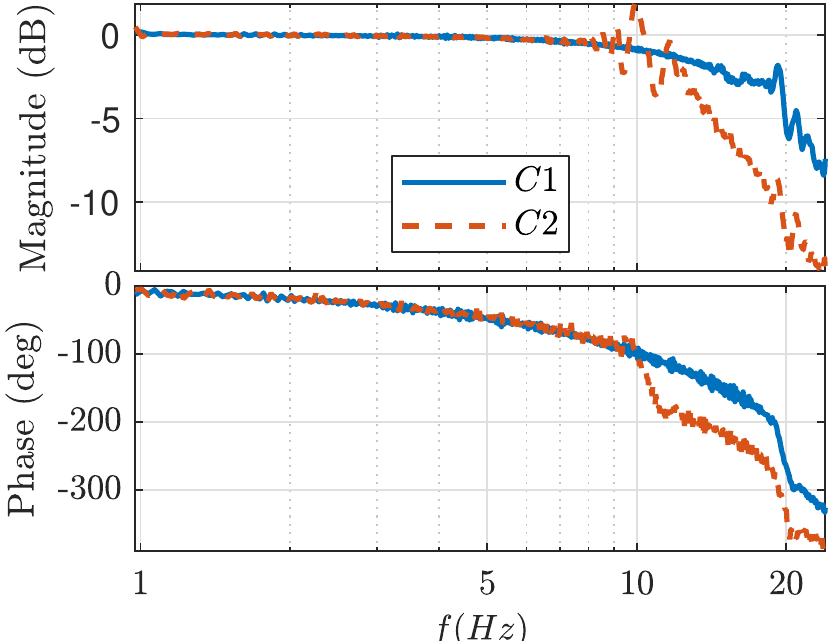}
	\caption{Magnitude (top) and phase (bottom) response of the actuator with controllers $C1$ (solid blue) and $C2$ (dashed orange).}
	\label{fig:frequencydomain_experiment_bode}
\end{figure}

\begin{figure}[!h]
	\centering
	\includegraphics[width=0.85\columnwidth]{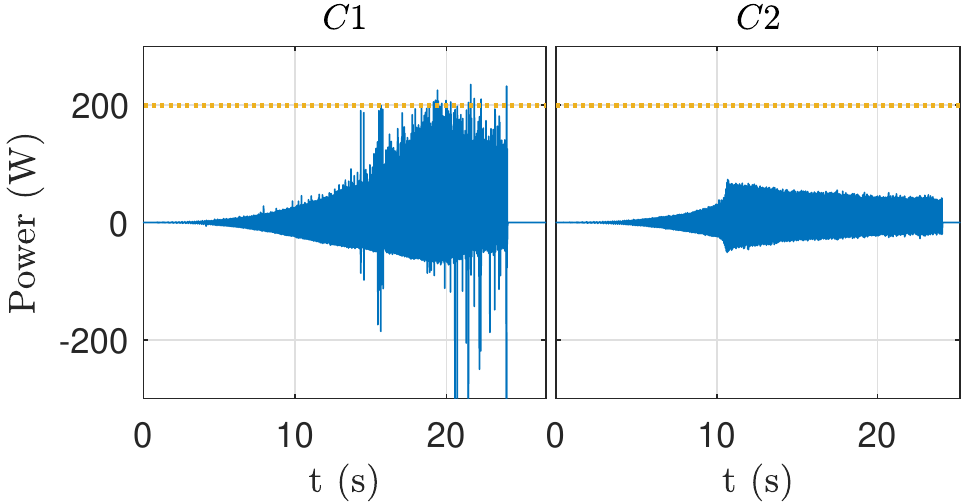}
	\caption{Power consumption of controllers $C1$ (left) and $C2$ (right) under power supply limit (dotted yellow) during the chirp test.}
	\label{fig:frequencydomain_experiment_power}
\end{figure}

In conclusion, both time-domain and frequency-domain experiments on a simple 1-DoF actuator provide preliminary evidence verifying theoretical advantages of the main idea in this paper, which is to use exact model of power supply limit in controllers instead of the approximate model.

\section{Conclusions and Future Works}
\label{sec:conclusion} 
In this work, motion control systems
with strict power limits are considered. Effects of these physical limits
as being a nonlinearity on dynamics are analysed rigorously in both frequency and time domains. In the former analysis, our exploration of the frequency response behavior of power supply limits follows from derivations of describing function and maximum bandwidth frequency. Time domain investigation is done with reference to controller design. In particular, we consider dynamic classical controllers for linear and nonlinear Euler-Lagrange systems. Finally, two general control architectures are proposed to improve control performance by leveraging the power supply limit through dynamic power allocation for multi DoF systems. In particular, these architectures are developed in the framework of constrained optimal control and constrained inverse optimal control where QP problems are formulated for linear systems and nonlinear systems, respectively. Finally, preliminary results of an experiment are presented and briefly discussed to support the theoretical contributions in the paper.

Since the presented material provides a new perspective on an old problem which has been becoming more important with recent developments in mobile systems, there are various future work directions : Analysis on stability and performance limitations due to power limit can be geared towards underactuated nonlinear systems for which swing-up energy control is a commonly used strategy \cite{spong.IEEECSM1995,astrom_furuta.IFAC1996}. The connection between passivity-based control and power supply limits might be interesting to explore further since passivity-based control is related to regulation of system energy. Convexification of the unified controller which solves both dynamic power allocation and control problems for linear mechanical systems by a finite-horizon QP problem in conjuction with power supply limit constraints is also an interesting avenue to explore since it would expand the application areas of those controllers by reducing computational requirements. Finally, further experimental results are needed to verify both classical controllers considering power supply limits and the unified algorithm embedding dynamic power allocation and feedback control.

\bibliography{references}

\begin{thebibliography}{10}
\expandafter\ifx\csname url\endcsname\relax
  \def\url#1{\texttt{#1}}\fi
\expandafter\ifx\csname urlprefix\endcsname\relax\def\urlprefix{URL }\fi
\expandafter\ifx\csname href\endcsname\relax
  \def\href#1#2{#2} \def\path#1{#1}\fi

\bibitem{melhuish.CONF2000}
I.~Kelly, O.~Holland, C.~Melhuish, {SlugBot: A Robotic Predator in the Natural
  World}, in: Proc. of the 5th Int. Symposium on Artificial Life and Robotics,
  2000.

\bibitem{landis.PVSC2005}
G.~A. Landis, {Exploring Mars with Solar-Powered Rovers}, in: Conf. Record of
  the 31st IEEE Photovoltaic Specialists Conf., 2005.

\bibitem{manning.SYSoSE2013}
R.~Welch, D.~Limonadi, R.~Manning, {Systems Engineering the Curiosity Rover: A
  Retrospective}, in: Proc. of the 8th Int. Conf. on System of Systems
  Engineering, 2013.

\bibitem{hutter.CONF2017}
H.~Kolvenbach, M.~Hutter, {Life Extension: An Autonomous Docking Station for
  Recharging Quadrupedal Robots}, in: Proc. of the 11th Int. Conf. on Field and
  Service Robotics, 2017.

\bibitem{avik_koditschek.RAL2016}
G.~D. Kenneally, A.~De, D.~E. Koditschek, {Design Principles for a Family of
  Direct-Drive Legged Robots}, IEEE Robotics and Automation Letters 1~(2)
  (2016) 900--907.

\bibitem{hurst_hatton.RSS2016}
A.~Abate, J.~W. Hurst, R.~L. Hatton, {Mechanical Antagonism in Legged Robots},
  in: Proc. of the Robotics : Science and Systems, 2016.

\bibitem{hatton_hurst.IEEE2018}
S.~Rezazadeh, A.~Abate, R.~L. Hatton, J.~W. Hurst, {Robot Leg Design : A
  Constructive Framework}, IEEE Access 6 (2018) 54369 -- 54387.

\bibitem{collins_ruina_tedrake_wisse.SCIENCE2005}
S.~Collins, A.~Ruina, R.~Tedrake, M.~Wisse, {Efficient Bipedal Robots Based on
  Passive-Dynamic Walkers}, Science 307~(5712) (2005) 1082--1085.

\bibitem{secer_saranli.TRO2018}
G.~Secer, U.~Saranli, {Control of Planar Spring-Mass Running Through Virtual
  Tuning of Radial Leg Damping}, IEEE Trans. on Robotics 34~(5) (2018)
  1370--1383.

\bibitem{sangbae.IROS2012}
S.~Seok, A.~Wang, D.~Otten, S.~Kim, {Actuator Design for High Force
  Proprioceptive Control in Fast Legged Locomotion}, in: Proc. of the IEEE/RSJ
  Int. Conf. on Intelligent Robots and Systems, 2012.

\bibitem{doyle.ACC1987}
J.~C. Doyle, R.~S. Smith, D.~F. Enns, {Control of Plants with Input Saturation
  Nonlinearities}, in: Proc. of the American Control Conf., 1987.

\bibitem{choi_lee.TIE2009}
J.~W. Choi, S.~C. Lee, {Antiwindup Strategy for PI-Type Speed Controller}, IEEE
  Trans. on Industrial Electronics 56~(6) (2009) 2039--2046.

\bibitem{ross.ISA1967}
H.~A. Fertik, C.~W. Ross, {Direct Digital Control Algoritm With Anti-Windup
  Feature}, ISA Transactions 6~(4) (1967) 317--328.

\bibitem{hanus.AUTO1987}
R.~Hanus, M.~Kinnaert, J.~L. Henrotte, {Conditioning Technique : A General
  Anti-Windup and Bumpless Transfer Method}, Automatica 23 (1987) 729--739.

\bibitem{kolmanovsky.AUTO2002}
E.~Gilbert, I.~Kolmanovsky, {Nonlinear Tracking Control in the Presence of
  State and Control Constraints : A Generalized Reference Governor}, Automatica
  38~(12) (2002) 2063--2073.

\bibitem{kolmanovsky.ACC2014}
I.~Kolmanovsky, E.~Garone, S.~D. Cairano, {Reference and Command Governors : A
  Tutorial on Their Theory and Automative Applications}, in: Proc. of the
  American Control Conf., 2014.

\bibitem{hall.TIE2001}
A.~S. Hodel, C.~E. Hall, {Variable-Structure {{PID}} Control to Prevent
  Integrator Windup}, IEEE Trans. on Industrial Electronics 48~(2) (2001)
  442--451.

\bibitem{zaccarian.EJC2009}
S.~Galeani, S.~Tarbouriech, M.~Turner, L.~Zaccarian, {A Tutorial on Modern
  Anti-Windup Design}, European Journal of Control 3~(4) (2009) 418--440.

\bibitem{astrom.ACC1989}
K.~J. Astrom, L.~Rundqwist, {Integrator Windup and How to Avoid It?}, in: Proc.
  of the American Control Conf., 1989.

\bibitem{teel_zaccarian.TAC2003}
G.~Grimm, J.~Hatfield, I.~Postlethwaite, A.~R. Teel, M.~Turner, L.~Zaccarian,
  {Anti-Windup for Stable Linear Systems with Input Saturation : An LMI-Based
  Synthesis}, IEEE Trans. on Automatic Control 48~(9) (2003) 1509--1525.

\bibitem{morari.AUTO2001}
E.~F. Mulder, M.~V. Kothare, M.~Morari, {Multivariable Anti-Windup Controller
  Synthesis Using Linear Matrix Inequalities}, Automatica 37~(9) (2001)
  1407--1416.

\bibitem{tarbouriech.CDC1997}
C.~Pittet, S.~Tarbouriech, C.~Burgat, {Stability Regions for Linear Systems
  with Saturating Controls via Circle and Popov Criteria}, in: Proc. of the
  36th Conf. on Decision and Control, 1997.

\bibitem{boyd.CDC1998}
H.~Hindi, S.~Boyd, {Analysis of Linear Systems with Saturation using Convex
  Optimization}, in: Proc. of the 37th IEEE Conf. on Decision and Control,
  1998.

\bibitem{forsyth.TIE2015}
D.~Wu, R.~Todd, A.~J. Forsyth, {Adaptive Rate-Limit Control for Energy Storage
  Systems}, IEEE Trans. on Industrial Electronics 62~(7) (2015) 4231--4240.

\bibitem{gelb.BOOK1968}
A.~Gelb, W.~E.~V. Velde, {Multiple-Input Describing Functions and Nonlinear
  System Design}, McGraw-Hill, 1968.

\bibitem{li.IEEE2018}
H.~Li, J.~Shang, B.~Zhang, X.~Zhao, N.~Tan, C.~Liu, {Stability Analysis With
  Considering the Transition Interval for PWM DC-DC Converters Based on
  Describing Function Method}, IEEE Access 6 (2018) 48113--48124.

\bibitem{boiko.IJC2009}
I.~Boiko, {Analysis of Sliding Modes in the Frequency Domain}, International
  Journal of Control 50~(9) (2009) 1442--1446.

\bibitem{nien.TPOWELEC2016}
S.~F. Hsiao, D.~Chen, C.~J. Chen, H.~S. Nien, {A New Multiplefrequency
  Small-Signal Model for High-Bandwidth Computer V-core Regulator
  Applications}, IEEE Transactions on Power Electronics 31~(1) (2016) 733--742.

\bibitem{fridman.AUTO2019}
U.~Perez-Ventura, L.~Fridman, {Design of Super-Twisting Control Gains : A
  Describing Function Based Methodology}, Automatica 99 (2019) 175--180.

\bibitem{munkres.BOOK2000}
J.~Munkres, {Section 26 in Topology}, 2nd Edition, Pearson, 2000.

\bibitem{reyhanoglu.TAC1999}
M.~Reyhanoglu, A.~van~der Schaft, N.~Mcclamroch, I.~Kolmanovsky, {{Dynamics and
  Control of a Class of Underactuated Mechanical Systems}}, {IEEE} Trans. on
  Automatic Control 44~(9) (1999) 1663--1671.
\newblock \href {https://doi.org/10.1109/9.788533}
  {\path{doi:10.1109/9.788533}}.

\bibitem{ortega.BOOK1998}
R.~Ortega, A.~Loria, P.~J. Nicklasson, H.~Sira-Ramirez, {{Passivity-Based
  Control of Euler-Lagrange Systems : Mechanical, Electrical, and
  Electromechanical Applications}}, Springer-Verlag, 1998.

\bibitem{mls.BOOK1994}
R.~M. Murray, Z.~Li, S.~S. Sastry, {Section 5.4 in A Mathematical Introduction
  to Robotic Manipulation}, CRC Press, 1994.

\bibitem{visioli.BOOK2006}
A.~Visioli, {Section 3.3.2 in Practical PID Control}, Springer, 2006.

\bibitem{zaccarian_teel.BOOK2011}
L.~Zaccarian, A.~R. Teel, {Section 2.3.5 in Modern Anti-windup Synthesis},
  Princeton University Press, 2011.

\bibitem{lin.AUTO2002}
T.~Hu, Z.~Lin, B.~M. Chen, {An Analysis and Design Method for Linear Systems
  Subject to Actuator Saturation and Disturbance}, Automatica 38~(2) (2002)
  351--359.
\newblock \href {https://doi.org/10.1016/s0005-1098(01)00209-6}
  {\path{doi:10.1016/s0005-1098(01)00209-6}}.

\bibitem{lin.AUTO2004}
C.~Lin, Q.~Wang, T.~H. Lee, {{An Improvement on Multivariable {PID} Controller
  Design via Iterative {LMI} Approach}}, Automatica 40~(3) (2004) 519--525.

\bibitem{koditschek.TECHREP1989}
D.~E. Koditschek, {The Application of Total Energy as a Lyapunov Function for
  Mechanical Control Systems}, Tech. rep., University of Pennsylvania (1989).

\bibitem{tomei.TRO1991}
P.~Tomei, {Adaptive {{PD}} Controller for Robot Manipulators}, IEEE Trans.on
  Robotics and Automation 7~(4) (1991) 565--570.

\bibitem{khalil.BOOK2002}
H.~K. Khalil, {Section 4.2. Nonlinear Systems}, Prentice Hall, 2002.

\bibitem{boyd.BOOK2004}
S.~Boyd, L.~Vandenberghe, {Section 4.4 in Convex Optimization}, Cambridge
  University Press, 2004.

\bibitem{ortega.BOOK1990}
J.~M. Ortega, {Numerical Analysis : A Second Approach}, Vol.~3 of {Classics in
  Applied Mathematics}, SIAM, 1990.

\bibitem{freeman_kokotovic.ACC1995}
R.~Freeman, P.~Kokotovic, {Optimal Nonlinear Controllers for Feedback
  Linearizable Systems}, in: {Proc. of the American Control Conference}, 1995.

\bibitem{freeman_kokotovic.SIAM1996}
R.~A. Freeman, P.~V. Kokotovic, {Inverse Optimality in Robust Stabilization},
  {SIAM Journal on Control and Optimization} 34~(4) (1996) 1365--1391.

\bibitem{ames_grizzle_sreenath.TAC2014}
A.~D. Ames, K.~Galloway, J.~W. Grizzle, K.~Sreenath, Rapidly exponentially
  stabilizing control lyapunov functions and hybrid zero dynamics, {IEEE
  Transactions on Automatic Control} 59~(4) (2014) 876--891.

\bibitem{ames.ACC2016}
K.~Y. Chao, M.~J. Powell, A.~D. Ames, P.~Hur, {Unification of Locomotion
  Pattern Generation and Control Lyapunov Function-Based Quadratic Programs},
  in: Proc. of the American Control Conference, 2016.

\bibitem{lee.BOOK2012}
J.~M. Lee, {Chapter 9 in Introduction to Smooth Manifolds}, Springer, 2012.

\bibitem{primbs_doyle.AJC1999}
J.~A. Primbs, V.~Nevistic, J.~C. Doyle, {Nonlinear Optimal Control : A Control
  Lyapunov Function and Receding Horizon Perspective}, Asian Journal of Control
  1~(1) (1999) 14--24.

\bibitem{sreenath_ames_grizzle.IEEE2015}
K.~Galloway, K.~Sreenath, A.~D. Ames, J.~W. Grizzle, {Torque Saturation in
  Bipedal Robotic Walking Through Control Lyapunov Function-Based Quadratic
  Programs}, IEEE Access 3 (2015) 323--332.

\bibitem{tabuada_grizzle_ames.ACC2015}
A.~Mehra, W.~Ma, F.~Berg, P.~Tabuada, J.~W. Grizzle, A.~D. Ames, {Adaptive
  Cruise Control: Experimental Validation of Advanced Controllers on
  Scale-Model Cars}, in: Proc. of the American Control Conf., 2015.

\bibitem{spong.IEEECSM1995}
M.~W. Spong, {The Swing Up Control for the Acrobot}, {IEEE Control Systems
  Magazine} 15~(1) (1995) 49--55.

\bibitem{astrom_furuta.IFAC1996}
K.~J. Astrom, K.~Furuta, Swinging up a pendulum by energy control, in: Proc. of
  the IFAC 13th World Congress, 1996.

\end{thebibliography}

\end{document}